\documentclass[aps,prx,floatfix,10pt,superscriptaddress,twocolumn,tightenlines,amsmath,amssymb,nofootinbib]{revtex4-1}

\usepackage[utf8]{inputenc}
\usepackage[T1]{fontenc}

\usepackage{graphicx}
\usepackage{amsthm}
\usepackage{hyperref}
\usepackage{braket}
\usepackage{comment}
\usepackage{palatino}
\usepackage{color}
\usepackage{psfrag}
\usepackage{ifsym}
\usepackage{amsfonts}
\usepackage{multirow}
\usepackage{soul}

\usepackage{bm}
\usepackage{bbm}
\usepackage{mathtools}

\usepackage[mathscr]{euscript}

\newcommand{\ketbra}[1]{ | #1 \rangle\!\langle #1 |}

\newcommand{\Tr} {\operatorname{Tr}}

\newcommand{\one}{\leavevmode\hbox{\small1\normalsize\kern-.33em1}}

\newcommand\id{\leavevmode\hbox{\small1\kern-3.3pt\normalsize1}}

\newcommand*{\proj}[1]{\ket{#1}\!\bra{#1}}
\newcommand{\I}{\mathbb{I}}

\newcommand{\CR}{r_C}
\newcommand{\CN}{n_C}

\newcommand{\NR}{\operatorname{NR}}

\newtheorem{remark}{Remark}

\newtheorem{theorem}{Theorem}
\newtheorem{corollary}{Corollary}
\newtheorem{lemma}{Lemma}

\newtheorem{observation}{Observation}

\newcommand{\change}[1]{#1}

\begin{document}
\title{Certification and Quantification of Multilevel Quantum Coherence}

\author{Martin Ringbauer}
\affiliation{Institute of Photonics and Quantum Sciences, School of Engineering and Physical Sciences, Heriot-Watt University, Edinburgh EH14 4AS, UK}
\affiliation{Centre for Engineered Quantum Systems, School of Mathematics and Physics, University of Queensland, Brisbane, Queensland 4072, Australia.}
\affiliation{Centre for Quantum Computation and Communication Technology, School of Mathematics and Physics, University of Queensland, Brisbane, Queensland 4072, Australia}

\author{Thomas R. Bromley}
\affiliation{Centre for the Mathematics and Theoretical Physics of Quantum Non-Equilibrium Systems (CQNE), School of Mathematical Sciences, The University of Nottingham, University Park, Nottingham NG7 2RD, United Kingdom}

\author{Marco Cianciaruso}
\affiliation{Centre for the Mathematics and Theoretical Physics of Quantum Non-Equilibrium Systems (CQNE), School of Mathematical Sciences, The University of Nottingham, University Park, Nottingham NG7 2RD, United Kingdom}

\author{Ludovico Lami}
\affiliation{Centre for the Mathematics and Theoretical Physics of Quantum Non-Equilibrium Systems (CQNE), School of Mathematical Sciences, The University of Nottingham, University Park, Nottingham NG7 2RD, United Kingdom}

\author{W.~Y.~Sarah Lau}
\affiliation{Centre for Engineered Quantum Systems, School of Mathematics and Physics, University of Queensland, Brisbane, Queensland 4072, Australia.}
\affiliation{Centre for Quantum Computation and Communication Technology, School of Mathematics and Physics, University of Queensland, Brisbane, Queensland 4072, Australia}

\author{Gerardo Adesso}
\email{gerardo.adesso@nottingham.ac.uk}
\affiliation{Centre for the Mathematics and Theoretical Physics of Quantum Non-Equilibrium Systems (CQNE), School of Mathematical Sciences, The University of Nottingham, University Park, Nottingham NG7 2RD, United Kingdom}

\author{Andrew G.~White}
\affiliation{Centre for Engineered Quantum Systems, School of Mathematics and Physics, University of Queensland, Brisbane, Queensland 4072, Australia.}
\affiliation{Centre for Quantum Computation and Communication Technology, School of Mathematics and Physics, University of Queensland, Brisbane, Queensland 4072, Australia}

\author{Alessandro Fedrizzi}
\email{a.fedrizzi@hw.ac.uk}
\affiliation{Institute of Photonics and Quantum Sciences, School of Engineering and Physical Sciences, Heriot-Watt University, Edinburgh EH14 4AS, UK}

\author{Marco Piani}
\email{marco.piani@strath.ac.uk}
\affiliation{SUPA and Department of Physics, University of Strathclyde, Glasgow G4 0NG, UK}

\begin{abstract}
\noindent Quantum coherence, present whenever a quantum system exists in a superposition of \change{multiple  classically distinct states}, marks one of the fundamental departures from classical physics. Quantum coherence has recently been investigated rigorously within a resource-theoretic formalism. However, the finer-grained notion of \emph{multilevel coherence}, which explicitly takes into account the number of superposed classical states, has remained relatively unexplored. A comprehensive analysis of multi-level coherence, which acts as the single-party analogue to multi-partite entanglement, is essential for understanding natural quantum processes as well as for gauging the performance of quantum technologies. Here we develop the theoretical and experimental groundwork for characterizing and quantifying multilevel coherence. \change{We prove that non-trivial levels of purity are required for multilevel coherence, as there is a ball of states around the maximally mixed state that do not exhibit multilevel coherence in any basis. We provide a simple necessary and sufficient analytical criterion to verify the presence of multilevel coherence, which leads to a complete classification of multilevel coherence for three-level systems. We present the robustness of multilevel coherence, a bona fide quantifier which we show to be numerically computable via semidefinite programming and experimentally accessible via multilevel coherence witnesses, which we introduce and characterize.} We further verify and lower-bound the robustness of multilevel coherence by performing a \change{semi-device-independent} phase discrimination task, which is implemented experimentally with four-level quantum probes in a photonic setup. Our results contribute to understanding the operational relevance of genuine multilevel coherence, also by demonstrating the key role it plays in enhanced phase discrimination---a primitive for quantum communication and metrology---and suggest new ways to reliably and effectively test the quantum behaviour of physical systems.
\end{abstract}

\maketitle

% =========================================================
% INTRODUCTION
% =========================================================
\noindent
Quantum coherence manifests whenever a quantum system is in a superposition of classically distinct states, such as different energy levels or spin directions. Formally, a quantum state displays coherence\footnote{For brevity, we will omit the qualifier ``quantum'' in the following.}, whenever it is described by a density matrix that is not diagonal with respect to the relevant orthogonal basis of classical states~\cite{streltsov2016quantum}. In this sense, coherence underpins virtually all quantum phenomena, yet has only recently been characterised formally~\cite{aberg2006quantifying,baumgratz2014quantifying,marvian2016quantify}. Coherence is now recognised as a fully-fledged resource and studied in the general framework of quantum resource theories~\cite{coecke2016mathematical,horodecki2013quantumness,brandao2015reversible,streltsov2016quantum}. This has led to a menagerie of possible ways to quantify coherence in a quantum system~\cite{baumgratz2014quantifying,Napoli2016,Piani2016,yuan2015intrinsic,winter2016operational,biswas2017interferometric,streltsov2015measuring,chitambar2016assisted,adesso2016measures,marvian2016quantum,girolami2014observable,streltsov2016quantum,ren2017quantitative}, along with an intense analysis of how coherence plays a role in fundamental physics, e.g., in quantum thermodynamics \cite{Kammerlander2015,UzdinPRX}, and in operational tasks relevant to quantum technologies, including quantum  algorithms and quantum metrology~\cite{Hillery2015,Napoli2016,Piani2016,marvian2016quantum,giorda2016coherence,zhang2016determining,ren2017quantitative,braun2017}.

Despite a great deal of recent progress, however, the majority of current literature focuses on a rather coarse-grained description of coherence, which is ultimately insufficient to reach a complete understanding of the fundamental role of quantum superposition in the aforementioned tasks. To overcome such limitations, one needs to take into consideration the number of classical states in coherent superposition---contrasted to the simpler question of whether any non-trivial superposition exists---which gives rise to the concept of \emph{multilevel} quantum coherence~\cite{levi2014quantitative,sperling2015convex,killoran2016converting}.
Similarly to the existence of different degrees of entanglement in multi-partite systems, going well beyond the mere presence or absence of entanglement and corresponding to different capabilities in quantum technologies \cite{horodecki2009quantum,guhne2009entanglement}, one can then identify and study a rich structure for multilevel coherence. Deciphering this structure can yield a tangible impact on many areas of physics, such as condensed matter, statistical mechanics and transfer phenomena in many-body systems~\cite{levi2014quantitative,li2012witnessing,Lostaglio2015,Scholes2017}. For example, for understanding the role of coherence in the function of complex biological molecules, such as those found in light harvesting, it will be crucial to differentiate between pairwise coherence among the various sites in the molecule, and genuine multilevel coherence across many sites~\cite{Scholes2017,witt2013stationary,scholak2011efficient,tiersch2012critical}. In quantum computation, large superpositions of computational basis states need to be generated and effective benchmarking of such devices require proper tools to certify and quantify multilevel coherence.

Recent works have presented initial approaches to measuring the amount of multilevel coherence~\cite{chin2017generalized}, as well as schemes to convert it into bipartite and genuine multi-partite entanglement, enabling the fruitful use of entanglement theory tools  to study coherence itself~\cite{regula2017converting,killoran2016converting,chin2017conversion}. Nonetheless, an all-inclusive systematic framework for the characterization, certification, and quantification of multilevel coherence is still lacking.

Here we construct and present such a theoretical framework for multilevel coherence and apply it to the experimental verification and quantification of multilevel coherence in a quantum optical setting. We begin by developing a resource theory of multilevel coherence, in particular providing a new characterisation of the sets of multilevel coherence-free states (see Fig.~\ref{fig:hierarchy}a) and free operations, rigorously unfolding the hierarchy of multilevel coherence. \change{We present analytical criteria for multilevel coherence, which lead to a complete classification of multilevel coherence for three-level systems, and which establish lower bounds on the purity required to exhibit multilevel coherence.} We then formalise the robustness of multilevel coherence and show that it is an efficiently computable measure, which is experimentally accessible through multilevel coherence witnesses. Using photonic four-dimensional systems we demonstrate how to quantify, witness, and bound multilevel coherence experimentally. We prove that multilevel coherence, quantified by our robustness measure, has a natural operational interpretation as a fundamental resource for quantum phase discrimination~\cite{Napoli2016,Piani2016}, a cornerstone task for quantum \change{metrology} and communication technologies~\cite{gottesman1999fault,BarnettCroke}. In turn, we show how to exploit this task to experimentally lower-bound the robustness of multilevel coherence of an unknown quantum state in a \change{semi-device-independent} manner.

Our results yield a significant step forward in the theoretical and experimental quest for the full characterization of (multilevel) coherence as a core feature of quantum systems, and provide a practically useful toolbox for the performance assessment of upcoming quantum technologies exploiting multilevel coherence as a resource.

% =========================================================
% Results
% =========================================================
\section{Results}
\subsection{Resource theory of multilevel coherence}
We generalize the recently formalized resource theory of coherence~\cite{streltsov2016quantum} to the notion of multilevel coherence. We remind the reader that the general structure of a resource theory contains three main ingredients, which we present below: a set of \emph{free states}, which do not contain the resource, a set of \emph{free operations}, which are quantum operations that cannot create the resource, and a \emph{measure} of the resource.

\begin{figure}[t!]
  \begin{center}
  \includegraphics[width=0.8\columnwidth]{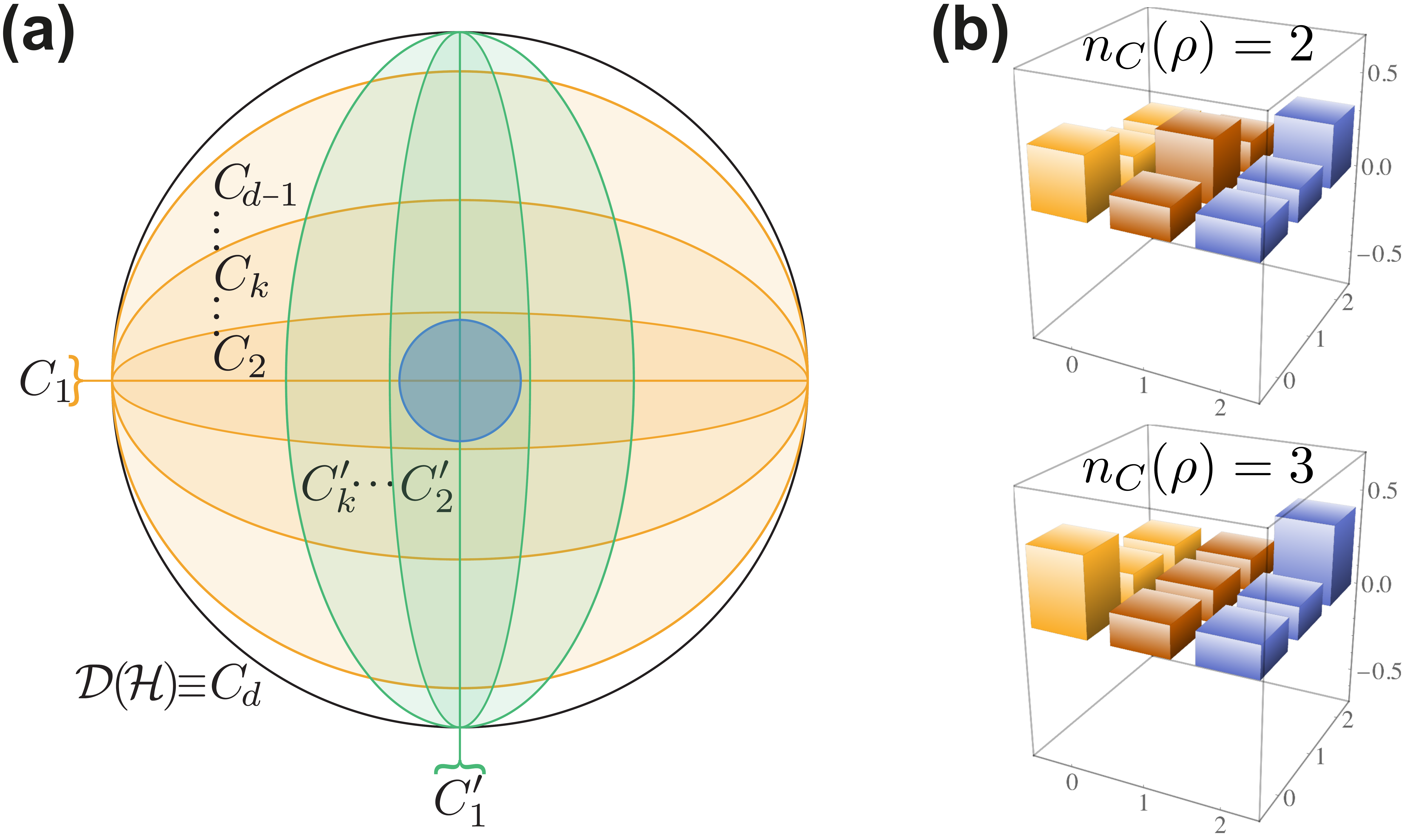}
\end{center}
    \vspace{-2em}
\caption{\textbf{Hierarchy of multilevel coherence.} \textbf{(a)} The set of states $\mathscr{D(H)}$ (outer circle) of a $d$-dimensional quantum system can be structured according to coherence number into the convex sets $C_{k}$ (orange shading) with $C_{1} \subset C_{2} \subset \ldots \subset C_{d}=\mathscr{D(H)}$. Note that a different choice of classical basis leads to a different hierarchy of sets $C_{k}'$ (green shading). However, as shown in the Supplementary Material~\cite{SI}, irrespective of the classical basis, there is a finite volume ball within $C_{2}$ (blue inner circle). This implies that, while almost all states exhibit some form of coherence, achieving genuine multilevel coherence is instead non-trivial and requires the state to be sufficiently far from the maximally mixed state.
\textbf{(b)} Real part of the density matrix of two example three-dimensional quantum states with equal off-diagonal elements, yet different multilevel-coherence properties. The upper state is a mixture of two level-coherent states of the form $\frac{1}{3}(\proj{\psi_{0,1}}+\proj{\psi_{0,2}}+\proj{\psi_{1,2}})$, where $\ket{\psi_{i,j}}=\frac{1}{\sqrt{2}}(\ket{i}+\ket{j})$, and thus has coherence number $n_C(\rho)=2$. \change{The lower state is a mixture of the maximally coherent state $(\ket{0}+\ket{1}+\ket{2})/\sqrt{3}$ with weight $1/2$ and of the incoherent-basis states $\ket{0}$ and $\ket{2}$, each with weight $1/4$. Every pure-state decomposition of the lower state must contain a superposition of all $\ket{0},\ket{1},\ket{2}$, as it can be verified by calculating numerically the robustness of three-level coherence $R_{C_2}\approx 0.0361$ (see Eq.~\ref{Eq:RMC}) or by using the comparison matrix criterion of Section~\ref{sec:conditionsgenuinemultilevel}}. Hence, although both states exhibit the same off-diagonal elements, only the lower state has genuine multilevel coherence, requiring experimental control that is coherent across multiple levels. This exemplifies the fine-grained classification of coherence and of experimental capabilities that studying multilevel coherence provides.}
\label{fig:hierarchy}
\end{figure}

\textbf{Multilevel coherence-free quantum states.}
Consider a $d$-dimensional quantum system with Hilbert space $\mathscr{H}\simeq\mathbb{C}^d$, spanned by an orthonormal basis $\{\ket{i }\}_{i=1}^{d}$, with respect to which we measure quantum coherence. The choice of classical basis is typically fixed to correspond to the eigenstates of a physically relevant observable like the system Hamiltonian. Any pure state $\ket{\psi} \in \mathscr{H}$ can be written in this basis as $\ket{\psi} = \sum_{i=1}^{d} c_{i} \ket{i}$ with $\sum_{i=1}^{d}|c_{i}|^{2}=1$. The state $\ket{\psi}$ exhibits \change{\emph{some}} quantum coherence with respect to the basis $\{\ket{i }\}_{i=1}^{d}$ whenever at least two of the coefficients $c_{i}$ are non-zero~\cite{streltsov2016quantum}. The \change{\emph{multilevel}} nature of coherence is revealed by the number of non-zero coefficients $c_{i}$, the \emph{coherence rank}~$\CR$~\cite{witt2013stationary,levi2014quantitative}. We say that a state $\ket{\psi}$ has coherence rank $\CR(\ket{\psi}) = k$ if exactly $k$ of the coefficients $c_{i}$ are non-zero. The notion of coherence rank thereby provides a \emph{fine-grained} account of the quantum coherence of $\ket{\psi}$, \change{as compared to merely establishing the presence of some coherence}.

To generalise multilevel coherence to mixed states $\rho \in \mathscr{D(H)}$, we define the sets $C_{k} \subseteq \mathscr{D(H)}$ with $k \in \{1,\ldots,d\}$, given by all probabilistic mixtures of pure density operators $\proj{\psi}$ with a coherence rank of at most $k$,
\begin{equation}
C_k := \operatorname{conv}\{ \proj{\psi} : \CR(\ket{\psi})\leq k \text\},
\label{Ck}
\end{equation}
where $\operatorname{conv}$ stands for `convex hull'. $C_{1}$ is the set of fully incoherent states, given by density matrices that are diagonal in the classical basis, while $C_{d} \equiv \mathscr{D(H)}$ is the set of all states. The intermediate sets obey the strict hierarchy, $C_{1} \subset C_{2} \subset \ldots \subset C_{d}$ (see Fig.~\ref{fig:hierarchy}a) and are the free states in the resource theory of multilevel coherence, e.g.\ $C_k$ is the set of $(k+1)$-level coherence-free states.

For a general mixed state one defines the coherence number~$\CN$~\cite{chin2017conversion,chin2017generalized,regula2017converting}, such that a state $\rho \in \mathscr{D(H)}$ has a coherence number $\CN(\rho)=k$ if $\rho \in C_{k}$ and $\rho \notin C_{k-1}$ (for consistency, we set $C_0 = \emptyset$). This parallels the notions of Schmidt number~\cite{terhal2000schmidt} and entanglement depth~\cite{sorensen_entanglement_2001} in entanglement theory. A state with coherence number $\CN(\rho)=k$ can be decomposed into (at most $d^2$) pure states with coherence rank at most $k$, while every such decomposition must contain at least one state with coherence rank at least $k$. A state with $\CN(\rho)=k$ is said to exhibit \emph{genuine} $k$-level coherence, distinguishing it from states that \change{may display coherence between several pairs of levels -- potentially even between all such pairs -- yet can be prepared as mixtures of pure states with relatively lower-level coherence, see Fig.~\ref{fig:hierarchy}b}.
In an experiment, the presence of multilevel coherence proves the ability to coherently manipulate a physical system across many of its levels, much in the same way that the creation of states with large entanglement depth provides a certification of the coherent control over several systems.

Note that a state may, at the same time, display large tout-court coherence, but have vanishing higher-level coherence. This is the case, for example, for a superposition of $d-1$ basis elements, like $\ket{\psi}=\sqrt{1/(d-1)}\sum_{i=2}^d \ket{i}$, which does not display $d$-level coherence despite being highly coherent. On the other hand, a pure state may be arbitrarily close to one of the elements of the incoherent basis, yet display non-zero genuine multilevel coherence for all $k$. This is the case, for example, for the state $\ket{\phi}=\sqrt{1-\epsilon}\ket{1}+\sqrt{\epsilon/(d-1)}\sum_{i=2}^d\ket{i}$ for small $\epsilon$.
It should be clear that the above multi-level classification provides a much finer description of the coherence properties of quantum systems, but that it is also important to elevate such a finer qualitative classification to a finer quantitative description, as we will do in the following, specifically in Section

\textbf{Multilevel coherence-free operations and $k$-decohering operations.}
The second ingredient in the resource theory of multilevel coherence is the set of operations that do not create multilevel coherence. A general quantum operation $\Lambda$ is described by a linear completely-positive and trace-preserving (CPTP) map, whose action on a state $\rho$ can be written as $\Lambda(\rho) = \sum_{i} K_{i}\rho K_{i}^{\dagger}$, in terms of (non-unique) Kraus operators $\{K_{i}\}$ with $\sum_i K_i^\dagger K_i = \I$~\cite{nielsen2010quantum}. For any map $\Lambda$ and any set $S$ of states, we denote $\Lambda(S)\coloneqq\{\Lambda(\rho):\rho\in S\}$. Generalising the formalism introduced for standard coherence~\cite{aberg2006quantifying,baumgratz2014quantifying,streltsov2016quantum}, we refer to a CPTP map $\Lambda$ as a \emph{$k$-coherence preserving} operation if it cannot increase the coherence level, i.e.\ $\Lambda(C_{k})\subseteq C_k$. An important subset of these are the \emph{$k$-incoherent operations}, which are all CPTP maps for which there exists a set of Kraus operators $\{K_i\}$ such that $K_{i}\rho K_{i}^{\dagger}/\mbox{Tr}(K_{i}\rho K_{i}^{\dagger}) \in C_{k}$ for any $\rho \in C_{k}$ and all $i$. Note that the (fully) incoherent operations from the resource theory of coherence correspond to $k=1$. In the Supplementary Material~\cite{SI} we prove that fully incoherent operations are also $k$-incoherent operations for all $k$, and we further define the notion of \emph{$k$-decohering} maps as those that destroy multilevel coherence\change{: an operation $\Lambda$ is $k$-decohering if $\Lambda(\mathscr{D(H)})\subseteq C_k$}. \change{In particular, we introduce a family of maps that generalize the fully decohering map $\Delta[X]=\sum_{i=1}^d \proj{i}X\proj{i}$, which is such that $\Delta(\mathscr{D(H)})\equiv C_1$.}

\textbf{Measure of multilevel coherence.}
The final ingredient for the resource theory of multilevel coherence is a well-defined measure. Very few quantifiers of such a resource have been suggested, and those that exist lack a clear operational interpretation \cite{chin2017conversion,chin2017generalized}. Furthermore, many of the quantifiers of coherence, such as the intuitive $l_{1}$ norm of coherence, which measures the off-diagonal contribution to the density matrix, fail to capture the intricate structure of multilevel coherence, as indicated in Fig.~\ref{fig:hierarchy}b. Here we introduce the \emph{robustness of multilevel coherence} (RMC) $R_{C_k}(\rho)$ as a bona-fide measure that is directly accessible experimentally and efficient to compute for any density matrix. The robustness of $(k{+}1)$-level coherence can be understood as the minimal amount of noise that has to be added to a state to destroy all $(k{+}1)$-level coherence, defined as
\begin{equation}\label{Eq:RMC}
R_{C_k}(\rho) := \inf_{\tau\in \mathscr{D(H)}}\left\lbrace s \geq 0 : \frac{\rho + s \tau}{1+s} \in C_k \right\rbrace .
\end{equation}
This measure generalises the recently introduced robustness of coherence~\cite{Napoli2016,Piani2016} (corresponding to $R_{C_1}(\rho)$) to provide full sensitivity to the various levels of multilevel coherence. As a special case of the general notion of \emph{robustness} of a quantum resource~\cite{vidal1999robustness,steiner2003generalized,piani2015necessary,geller2014quantifying,harrow2003robustness,brandao2005quantifying,brandao2006witnessed,bu2017asymmetry}, the quantities $R_{C_k}$ are known to be valid resource-theoretic measures~\cite{horodecki2009quantum,streltsov2016quantum}, satisfying non-negativity, convexity, and monotonicity on average with respect to stochastic free operations~\cite{vidal1999robustness,steiner2003generalized,Napoli2016,Piani2016,chiribella2016optimal}. The latter means for any $\rho$ that $R_{C_{k}}(\rho) \geq  \sum_{i} p_{i} R_{C_{k}}(\rho_{i})$ for all $k$-incoherent operations with Kraus operators $\{K_{i}\}$ such that $p_{i} = \mbox{Tr}(K_{i}\rho K_{i}^{\dagger})$ and $\rho_{i} = K_{i}\rho K_{i}^{\dagger}/p_{i}$. Since (fully) incoherent operations are $k$-incoherent for any $k$, the RMC also satisfies the \emph{strict monotonicity} requirement for coherence measures~\cite{baumgratz2014quantifying,streltsov2016quantum}, see Supplementary Material~\cite{SI}.

Crucially, we find that the RMC can be posed as the solution of a semidefinite program (SDP) optimization problem~\cite{boyd2004convex,vandenberghe1996semidefinite,watrousSDP}, see Supplementary Material~\cite{SI}. A variety of algorithms exists to solve SDPs efficiently~\cite{vandenberghe1996semidefinite}, meaning that the RMC may be computed efficiently for any $k$---in stark contrast to the robustness of entanglement~\cite{vidal1999robustness,steiner2003generalized} where one has to deal with the subtleties of the characterisation of the set of separable states~\cite{brandao2005quantifying}. For an arbitrary $d$-dimensional quantum state we find that
\begin{equation}
\label{Eq:RMCrange}
0 \leq R_{C_{k}}(\rho) \leq \frac{d}{k}-1 \quad \forall \rho \in \mathscr{D(H)}\, ,
\end{equation}
since any such state can be deterministically prepared using only (fully) incoherent operations~\cite{baumgratz2014quantifying} starting from the maximally coherent state $\ket{\psi^{+}_d}=d^{-1/2} \sum_{i=1}^d \ket{i}$, for which $R_{C_{k}}(\rho) = \frac{d}{k}-1$ (see Supplementary Material~\cite{SI}).

\subsection{Experimental verification and quantification of multilevel coherence}
We apply our theoretic framework to an experiment that produces four-dimensional quantum states with varying degree and level of coherence using the setup in Fig.~\ref{fig:Setup}. We use heralded single photons at a rate of $\sim10^4$Hz, generated via spontaneous parametric down-conversion in a $\beta$-Barium borate crystal, pumped by a femto-second pulsed laser at a wavelength of $410$~nm. We encode quantum information in the polarisation and path degrees of freedom of these photons to prepare $4$-dimensional systems~\cite{Ringbauer2015Epistemic} with the basis states $\ket{0}{=}\ket{H}_1, \ket{1}{=}\ket{V}_1, \ket{2}{=}\ket{H}_2, \ket{3}{=}\ket{V}_2$, where $\ket{p}_m$ denotes a state of polarisation $p$ in mode $m$. This dual-encoding allows for high-precision preparation of arbitrary pure quantum states of any dimension $d\leq4$ with an average fidelity of $\mathcal{F}=0.997\pm0.002$ and purity of $\mathcal{P}=0.995\pm0.003$. An arbitrary mixed state $\rho$ can be engineered as a proper mixture, by preparing the states of a pure-state decomposition of $\rho$ for appropriate fractions of the total measurement time and tracing out the classical information about which preparation was implemented. Using the same technique, we can also subject the input states to arbitrary forms of noise.

\begin{figure}[t!]
  \begin{center}
\includegraphics[width=0.9\columnwidth]{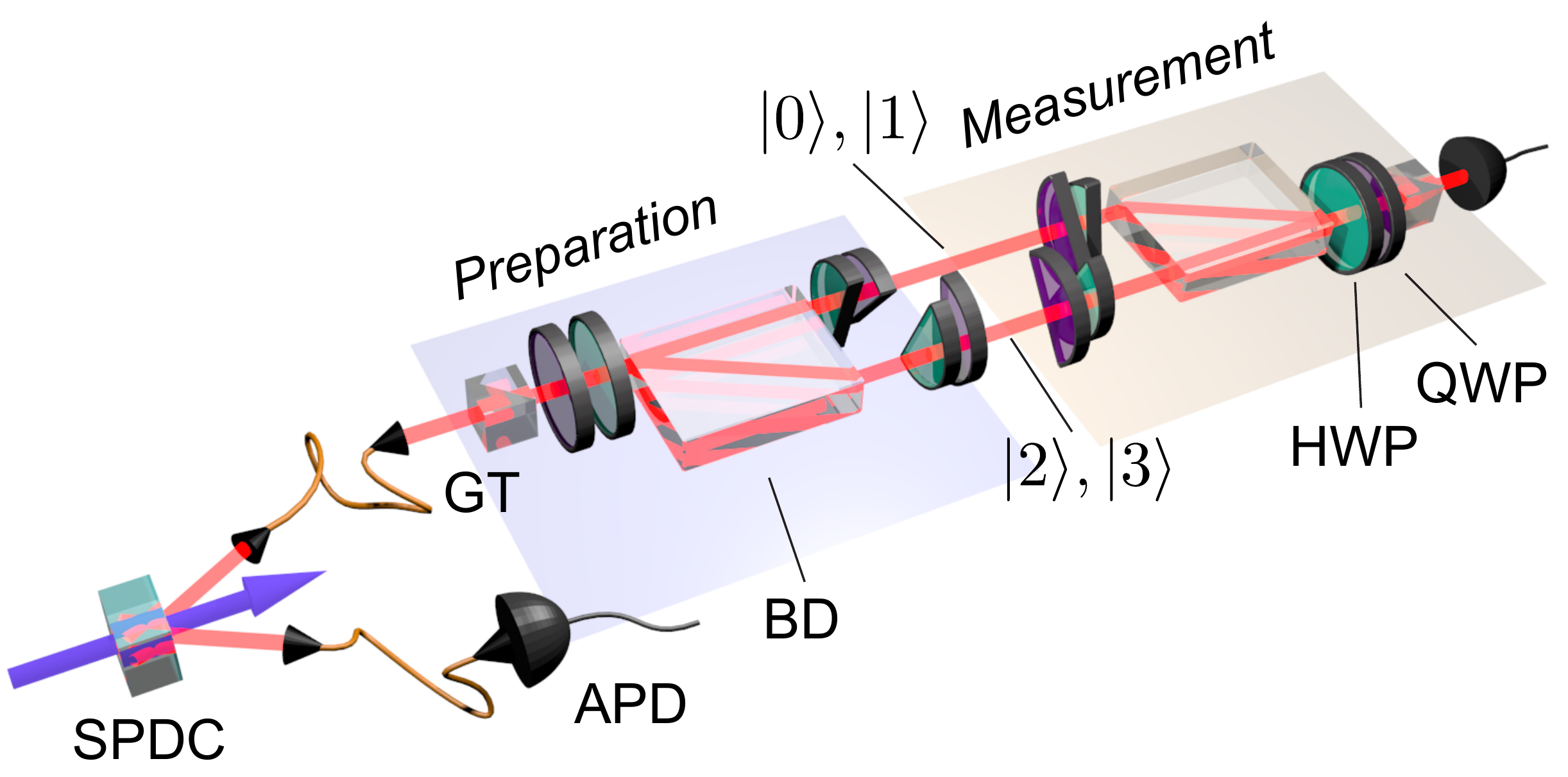}
  \end{center}
    \vspace{-1em}
\caption{\textbf{Experimental setup for probing multilevel coherence} in systems of dimension $d\leq 4$. Pairs of single photons are created via spontaneous parametric down-conversion (SPDC) in a BBO crystal, pumped at a wavelength of $410$~nm. The detection of one photon heralds the presence of the other, which is initialised in a horizontal polarisation state by means of a Glan-Taylor polariser (GT). A four-level quantum system is then prepared using polarisation encoding in each of the two spatial modes created by a calcite beam-displacer (BD). Three sets of half-wave (HWP) and quarter-wave plates (QWP) are used to control the amplitude and phase of the generated states. We prepare noisy maximally coherent states $\rho(p)$, Eq.~\eqref{Eq:NoisyMaximallyCoherent}, for several values of $p$ in dimension $d=4$, as detailed in (b). Arbitrary states can be prepared and measured in dimension $d\leq 4$ by manipulating only the corresponding subspaces.}
  \label{fig:Setup}
\end{figure}

Reversing the preparation stage of the setup allows us to implement arbitrary sharp projective measurements. Arbitrary generalized measurements~\cite{nielsen2010quantum} are correspondingly implemented as proper mixtures of a projective decomposition with an average fidelity of $\mathcal{F}=0.997\pm0.002$. By design, our experiment implements one measurement outcome at a time, which achieves superior precision through the use of a single fibre-coupling assembly~\cite{Ringbauer2015Epistemic}, while reducing systematic bias. The whole experiment is characterized by a quantum process fidelity of $\mathcal{F}_p=0.9956\pm0.0002$, limited by the interferometric contrast of ${\sim}300{:}1$. The latter is stable over the relevant timescales of the experiment due to the inherently stable interferometric design with common mode noise rejection for all but the piezo-driven rotational degrees of freedom of the second beam displacer. All data presented here was integrated over $20$s for each outcome, which is also much faster than the observed laser drifts on the order of hours. The main source of statistical uncertainties thus comes from the Poisson-distributed counting statistics. This has been taken into account through Monte-Carlo resampling with $10^4$ runs for tomographic measurements and $10^5$ runs for all other measurements. All experimental data presented in the figures and text throughout the manuscript are based on at least $10^5$ single photon counts and contain 5$\sigma$-equivalent statistical confidence intervals, which are with high confidence normal distributed unless otherwise stated.

\textbf{\change{Testbed family of states.}}
To illustrate the phenomenology of multilevel coherence, we consider a family of noisy maximally coherent states
\begin{equation}
\label{Eq:NoisyMaximallyCoherent}
\rho(p) = (1-p) \frac{\I}{d} + p \proj{\psi_d^{+}},
\end{equation}
with $p \in [0,1]$ and finite dimension $d$. These states interpolate between the maximally mixed state $\I/d$ (for $p=0$) and the maximally coherent state $\ket{\psi^{+}_d}=d^{-1/2} \sum_{i=1}^d \ket{i}$ (for $p=1$). For this class of states, the RMC can be evaluated analytically to be (see Supplementary Material~\cite{SI})
\begin{equation}
\label{Eq:ROCkNMC}
R_{C_k}(\rho(p)) = \max \left\lbrace \frac{p(d-1)-(k-1)}{k} , 0 \right\rbrace .
\end{equation}
In particular, this implies that $\rho(p)\in C_k$ for $p\leq \frac{k-1}{d-1}$ and $\rho(p) \notin C_k$ for $p> \frac{k-1}{d-1}$, see Supplementary Material~\cite{SI}. The family of noisy maximally coherent states thus provides the ideal testbed for our investigation, spanning the full hierarchy of multilevel coherence, see Fig.~\ref{fig:NMC}. Using the setup of Fig.~\ref{fig:Setup} we engineer noisy maximally coherent states $\rho(p)$ for $d=4$ and a variety of values of $p$. We then reconstruct the experimentally prepared states using maximum likelihood quantum state tomography and compute the robustness coherence for all $k$ by evaluating the corresponding SDP, Eq.~(S13) of the Supplementary Material. As illustrated in Fig.~\ref{Fig:NMCWitness}, this method produces very reliable results, however, it requires $d^2$ measurements and is thus experimentally infeasible already for medium-scale systems. \change{In the following we introduce and use multilevel-coherence witnesses and other techniques to overcome such a limitation.}

\begin{figure}[t!]
  \begin{center}
\includegraphics[width=0.9\columnwidth]{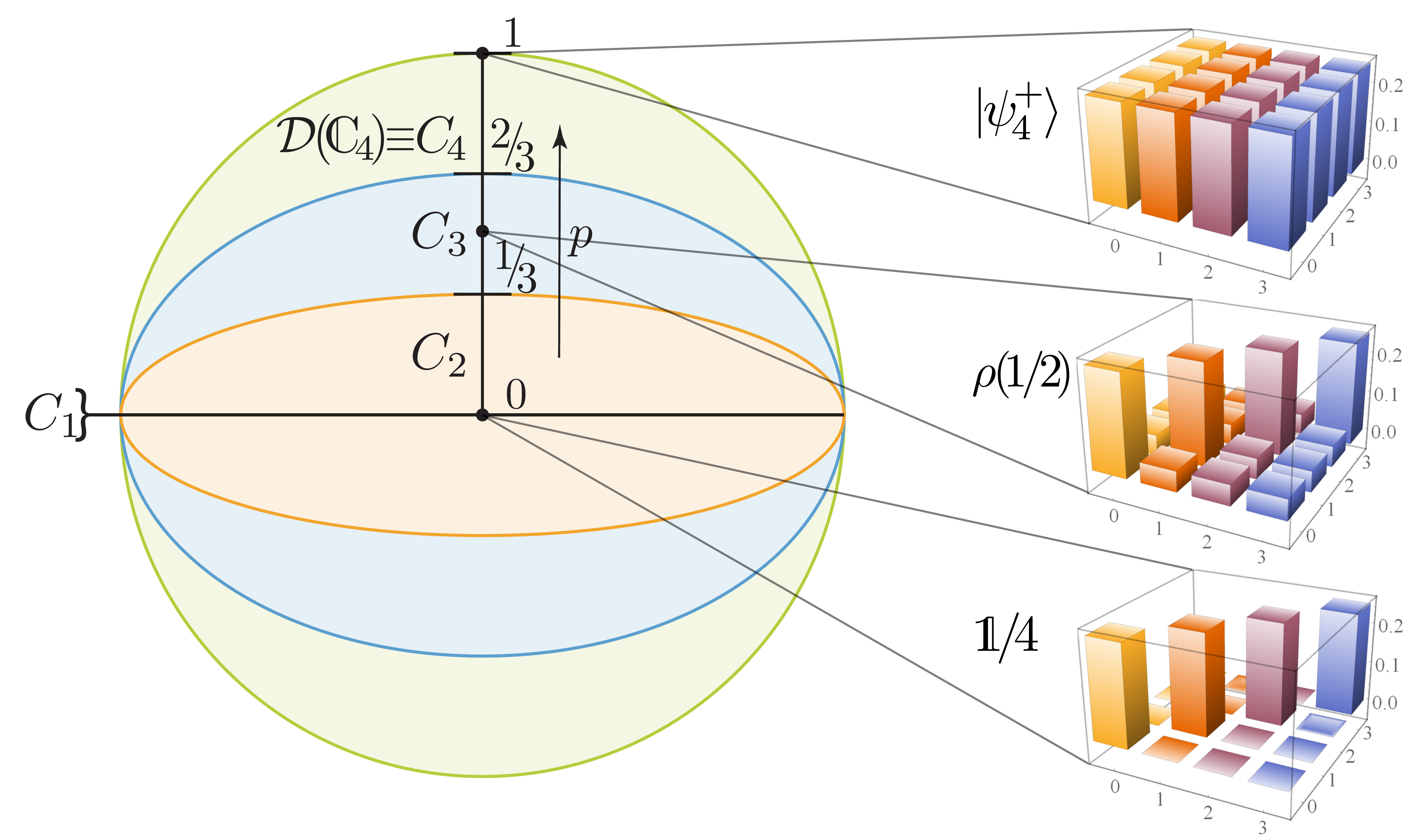}
  \end{center}
    \vspace{-2em}
\caption{\textbf{Multilevel coherence in $4$-dimensional noisy maximally coherent states.} With varying parameter $p \in [0,1]$, the coherence number of a $4$-dimensional noisy maximally coherent state ranges from $\CN(\rho(p)){=}1$ for $p{=}0$, to $\CN(\rho(p)){=}2$ for $p\in ]0,\frac{1}{3}]$ (blue region), $\CN(\rho(p)){=}3$ for $p\in ]\frac{1}{3},\frac{2}{3}]$ (orange region), and $\CN(\rho(p)){=}4$ for $p\in ]\frac{2}{3},1]$ (green region). We use this colour scheme throughout the paper to represent the $3$ non-trivial levels of coherence in a $4$-dimensional system. On the right we show examples of ideal density matrices for $p=1$ (top), $p=1/2$ (middle), and $p=0$ (bottom).}
  \label{fig:NMC}
\end{figure}

\change{
\subsection{Conditions for genuine multilevel coherence}
\label{sec:conditionsgenuinemultilevel}
Given a density matrix $\rho$, it is immediate to decide whether $\CN(\rho)=1$, as, by definition, this happens if and only if $\rho$ is diagonal. While in Section~\ref{sec:witnessing} we show how one can witness any multilevel coherence through the use of tailored multilevel-coherence witnesses, in this section we focus on simple analytical necessary and sufficient criteria for multilevel coherence. Such criteria also allow us to establish that all sets $C_k$, for $k\geq 2$, have non-zero volume within the set of all states.

\textbf{Necessary and sufficient criteria for coherence beyond two levels.} Given a $d\times d$ matrix $A$, the associated \emph{comparison matrix} is defined as~\cite[Definition~2.5.10]{HJ2}
\begin{equation}
M(A)_{ij} = \left\{ \begin{array}{cc} |A_{ii}| & \text{if $i=j$,} \\ -|A_{ij}| & \text{if $i\neq j$ .} \end{array} \right.
\label{comparison matrix}
\end{equation}
We refer the reader to \cite[Section~2.5]{HJ2} for more details on the many properties of this construction. We now present our result on the full characterization of the set $C_2$ in arbitrary dimension, whose proof is given the Supplementary Material~\cite{SI}.

\begin{theorem} \label{thm comparison CN 2}
A density matrix $\rho$ is such that $\CN(\rho)\leq 2$ if and only if $M(\rho)\geq 0$ in the sense of positive semidefiniteness.
\end{theorem}

An easy corollary of the above result is a simple rule to completely classify qutrit states according to their coherence number. Namely, a qutrit state $\rho$ has coherence number at most $2$ if $\det M(\rho)\geq 0$, and $3$ otherwise~\cite{SI}.}

\change{
\textbf{Necessary conditions for multilevel coherence.}
As indicated in Fig.~\ref{fig:hierarchy}a, the set of fully incoherent states $C_{1}$ has zero volume within $\mathscr{D(H)}$~\cite{Piani2016}. This has the important consequence that a  state  generated randomly will practically never be fully incoherent, and that arbitrarily small perturbations applied to a fully incoherent state will create coherence~\cite{ferraro2010almost}. In other words, under realistic experimental conditions one cannot prepare or verify a fully incoherent state. In contrast, we show in the Supplementary Material~\cite{SI} that the sets $C_{k}$ are always of non-zero volume for any $k\geq 2$ and thus present a rich, and experimentally meaningful hierarchy within $\mathscr{D(H)}$, as shown in Fig.~\ref{fig:hierarchy}a.

Specifically, we have that, if a state $\rho$ satisfies
\begin{equation}
\rho \geq \frac{d-k}{d-1} \Delta(\rho),
\label{eq:sufficient for C_k}
\end{equation}
with $\Delta$ the fully decohering map, then $\rho\in C_k$. Furthermore, a corollary of Theorem \ref{thm comparison CN 2} is that if a state $\rho$ satisfies
\begin{equation}
\Tr(\rho^2)\leq\frac{1}{d-1},
\label{Eq:purity multilevel}
\end{equation}
then such a state cannot have multilevel coherence, i.e. $\rho\in C_2$ for any reference basis. Observe that the condition \eqref{Eq:purity multilevel} is equivalent to being close enough to the maximally mixed state $\I/d$~\cite{SI}, and that the upper bound in Eq.~\eqref{Eq:purity multilevel} is tight,
as it is achieved by states at the boundary of the set of density matrices, e.g., by $\rho=(\I-\proj{\psi})/(d-1)$, with $\ket{\psi}$ any arbitrary pure state~\cite{SI}.
This corollary can be considered the correspondent in coherence theory of the celebrated fact, in entanglement theory, that there is a ball of (fully) separable states surrounding the maximally mixed state~\cite{zyczkowski_volume_1998,gurvits_largest_2002,gurvits_separable_2003}.

}

\subsection{Witnessing multilevel coherence}
\label{sec:witnessing}
In analogy with the parallel concept for quantum entanglement, we introduce an efficient alternative to the tomographic approach: \emph{multilevel coherence witnesses}. \change{In the following we will denote by $\lambda^{\min}(X)$ and $\lambda^{\max}(X)$ the smallest and largest eigenvalues of a Hermitian operator/matrix $X=X^\dagger$, respectively.}

Since the sets $C_k$ are convex, for any $\rho \notin C_{k}$ there exists a $(k{+}1)$-level coherence witness $W$ such that $\mbox{Tr}(W \rho) < 0$ and $\mbox{Tr}(W \sigma) \geq 0$ for all $\sigma\in C_{k}$~\cite{rudin1991functional}. A negative expectation value for $W$ thus certifies the $(k+1)$-level coherence of $\rho$ in a single measurement.

Given any pure state $\ket{\psi} = \sum_{i=1}^{d}c_{i} \ket{i} \in \mathscr{H}$, one can construct a $(k{+}1)$-level coherence witness as
\begin{equation}
\label{eq:maxcohwitness}
W_k(\psi) = \I - \frac{1}{\sum_{i=1}^k |c_i^{\downarrow}|^2}\proj{\psi} ,
\end{equation}
where $c_{i}^{\downarrow}$ are the coefficients $c_{i}$ rearranged into non-increasing modulus order. \change{This construction ensures that $\bra{\phi}W\ket{\phi} \geq 0$ for all $\ket{\phi}$ with $\CR(\phi)\leq k$, since $\max_{\CR(\phi)\leq k} |\braket{\phi|\psi}|^2 = \sum_{i=1}^k |c_i^{\downarrow}|^2$}~\cite{SI}.  On the other hand, it is clear that $W_k(\psi)$ always reveals the $k{+}1$ coherence of $\ket{\psi}$ if present, since $\braket{\psi|W_k(\psi)|\psi} = 1- \left(\sum_{i=1}^{k}|c_{i}^{\downarrow}|^{2}\right)^{-1}$, which is negative if $\ketbra{\psi} \notin C_{k}$. For the maximally coherent state $\ket{\psi_d^+}$, we then find $W_k(\psi_d^+) = \I - \frac{d}{k} \proj{\psi_d^+}$.

\begin{figure}[t!]
  \begin{center}
\includegraphics[width=0.9\columnwidth]{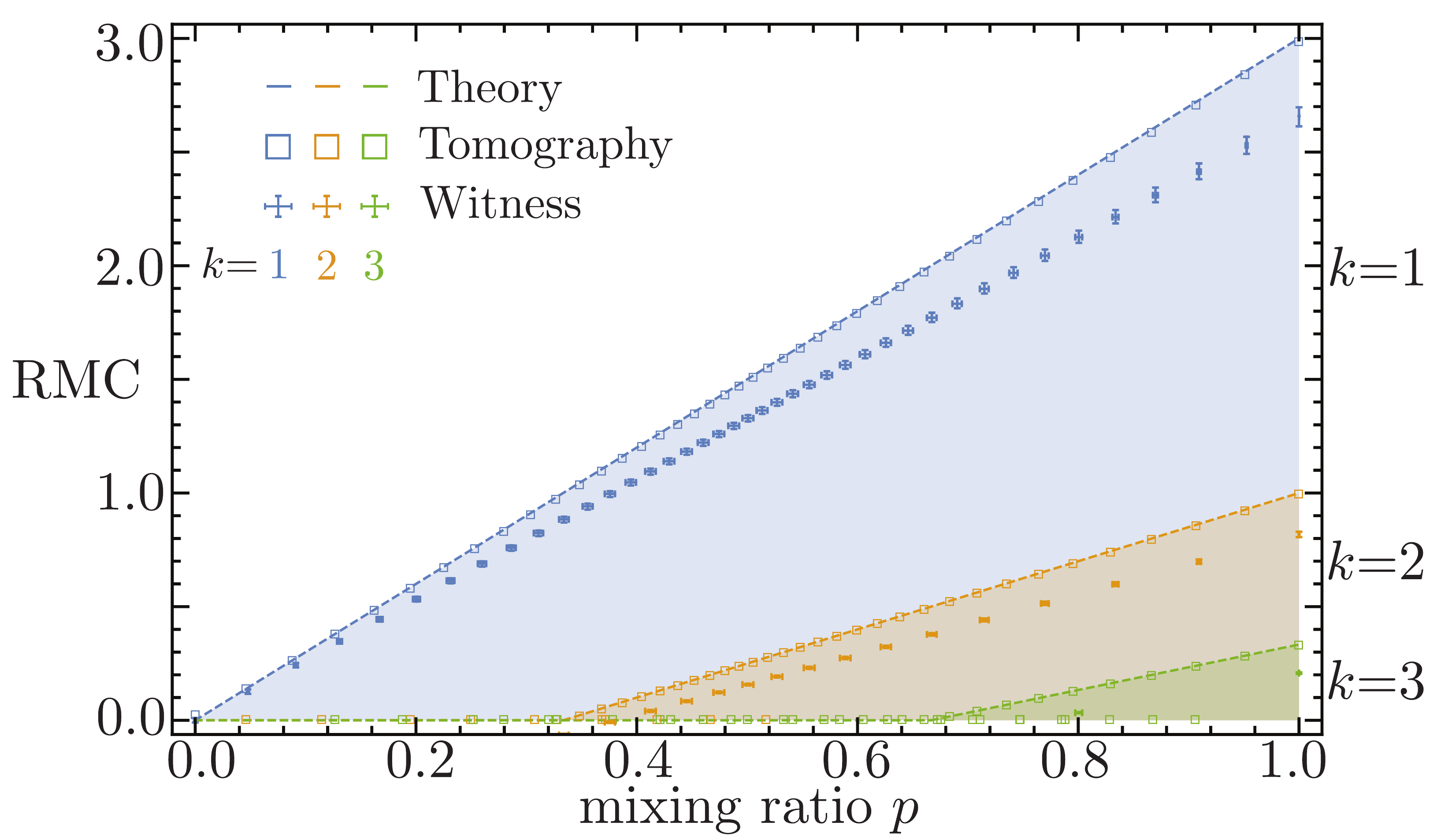}
  \end{center}
    \vspace{-2em}
\caption{\textbf{Measuring multilevel coherence.} The plot shows experimentally measured robustness of $(k+1)$-level coherence for a $4$-dimensional noisy maximally coherent state $\rho(p)$ as a function of $p \in [0,1]$. The solid lines represent the theory predictions from Eq.~\eqref{Eq:ROCkNMC} and the shaded areas indicate the regions where multilevel coherence for $k=1$ (blue), $k=2$ (orange), $k=3$ (green) can be observed. The open squares correspond the robustness of $(k+1)$-level coherence estimated from SDP in Eq.~\eqref{Eq:RobustnessDual} applied to the experimentally reconstructed density matrices. The 5$\sigma$ statistical confidence regions obtained from Monte-Carlo resampling are on the order of $10^{-3}$ for $p$ and on the order of $10^{-2}$ for the RMC. These are smaller than the symbol size and thus not shown. The data points with error bars correspond to the absolute values of the negative expectation values of $W_k(\psi_4^+)$ in Eq.~\eqref{eq:maxcohwitness}, which provide a lower bound on the RMC.}
  \label{Fig:NMCWitness}
\end{figure}

More generally, the set $C^*_k$ of $(k{+}1)$-level coherence witnesses is obtained as the dual of the set $C_k$ and is characterised by the following theorem, proved in the Supplementary Material~\cite{SI}.
\begin{theorem}\label{Theorem:WCharacterisation}
A self-adjoint operator $W$ is in $C_{k}^{*}$ if and only if
\begin{equation}\label{Eq:WitnessCharacterisation}
P_I W P_I \geq 0 \quad\forall I\in\mathcal{P}_k,
\end{equation}
where $\mathcal{P}_{k}$ is the set of all the $k$-element subsets of $\{1, 2, \ldots, d\}$, and $P_{I} \coloneqq \sum_{i \in I} \ket{i}\bra{i}$.
\end{theorem}
\noindent Hence, verifying that a given self-adjoint operator $W$ is a $(k{+}1)$-level coherence witness requires verifying the positive semidefiniteness of all $(k{\times}k)$-dimensional principal sub-matrices of the matrix representation of $W$ with respect to the classical basis.

\change{We observe that, while non-trivial multilevel-coherence witnesses necessarily have negative eigenvalues, the number of such negative eigenvalues is severely constrained~\cite{SI}. In particular, we have
\begin{observation}	
	\label{obs:negeig}
	A $(k+1)$-level coherence witness $W_k\in C^\star_k$ has at most $d-k$ negative eigenvalues. All the eigenvalues are bounded from below by $-\frac{d-k}{k}\lambda^{\max}(W_k)$.
\end{observation}
\noindent It is worth remarking that the eigenvector corresponding to the single negative eigenvalue of a $d$-level-coherence witness (that is, $k=d-1$) \emph{must} exhibit itself $d$-level coherence.}

\change{The characterization of multilevel-coherence witnesses of Theorem \ref{Theorem:WCharacterisation} finds explicit application in the dual form of the SDP formulation of the RMC~\cite{boyd2004convex}.} In the case of RMC strong duality holds, which means that the primal and dual forms of the problem are equivalent, with the latter given by
\begin{equation}\label{Eq:RobustnessDual}
\begin{array}{lll}
R_{{C}_k} (\rho) \,\,\, = \,\,\, & \max  \qquad & - \Tr(\rho W) \\[6pt]
& \textup{s.t.} & P_{I} W P_{I} \geq 0 \quad \forall I\in\mathcal{P}_k\\[6pt]
&  & W \leq \mathbb{I} \, . \\
\end{array}
\end{equation}
Hence, while a lower bound on $R_{{C}_k} (\rho)$ can be obtained from the negative expectation value of \emph{any} observable $W \in C_{k}^{\star}$ such that $W \leq \mathbb{I}$, \change{the dual SDP for the RMC actually computes an optimal $(k+1)$-level coherence witness whose expectation value matches $R_{{C}_k} (\rho)$.}

For the family of noisy maximally coherent state $\rho(p)$, the witness $W_k(\psi_d^+)$ of Eq.~\eqref{eq:maxcohwitness} turns out to be optimal, independently of the noise parameter $p$, and we calculate $\Tr(W_k(\psi_d^+) \rho(p)) = \frac{1}{k}[(k-1)-p(d-1)]$. Figure~\ref{Fig:NMCWitness} shows the absolute value of the experimentally obtained (negative) expectation values of $W_k(\psi_4^+)$ for a range of values of $p$. This demonstrates that multilevel coherence can be quantitatively witnessed in the laboratory using only a single measurement. Experimentally, however, implementing the optimal witness requires a projection onto a maximally coherent state, which is very sensitive to noise. Indeed, in our experiment we observed a small degree of beam steering by the wave plates, leading to phase uncertainty between the basis states $\ket{0,1}$ and $\ket{2,3}$. As a consequence, the witness becomes suboptimal and only provides a lower bound on the RMC of the experimental state. In contrast, our results show that the larger number of measurements in the tomographic approach and the associated maximum likelihood reconstruction add resilience to experimental imperfections.

\subsection{Bounding multilevel coherence}
\label{sec:bounding}
In practice, one might often neither be able to perform full tomography on a system, nor be able to measure the optimal witness. Remarkably, one can obtain a lower bound on the RMC of an experimentally prepared state $\rho$ from {\it any} set of experimental data. Specifically, the SDP in Eq.~(S16) in the Supplementary Material~\cite{SI} computes the minimal RMC of a state $\tau \in \mathcal{D(H)}$ that is consistent with a set of measured expectation values $o_{i} = \mbox{Tr}(O_{i}\rho)$ for $n$ observables $\{O_{i}\}_{i=1}^{n}$ to within experimental uncertainty. This is particularly appealing when one has already performed a set of (well-characterised) measurements and wishes to use these to estimate the multilevel coherence of the input state. Note that $d^2-1$ linearly independent observables (assuming vanishingly small errors, and not including the identity, which accounts for normalisation) are sufficient to uniquely determine the state, in which case we could use the original SDP, Eq.~\eqref{Eq:RobustnessDual}. \change{A similar approach can in principle be used to bound other quantum properties, like entanglement, from limited data~\cite{Guhne2007}, also via the use of SDPs~\cite{Eisert2007}. In the case of entanglement one still has to deal with the fact that the separability condition is not a simple SDP constraint, which is relevant even in the case of complete information: so, in general, the obstacle constituted by lack of information, combines with the obstacle of the difficulty of entanglement detection.}

We experimentally estimate the lower bounds from Eq.~(S16) in the Supplementary Material~\cite{SI} for an increasing number of randomly chosen observables $O_{i}$, measured on a $4$-dimensional maximally coherent state and on a noisy maximally coherent state with $p=0.8874\pm0.0007$, see Fig.~\ref{Fig:RMCObservables}. The results show that our lower bounds become non-trivial already for a small number of observables, and converge to a sub-optimal yet highly informative value. The remaining gap of about $5\%$ between these bounds and the tomographically estimated RMC is due to our conservative $5\sigma$ error bounds, and could be improved by incorporating maximum-likelihood or Bayesian estimation techniques, see Supplementary Material~\cite{SI} for details. We also find that the number of measurements required for non-trivial bounds increases slowly with the coherence level, and the bounds saturate more quickly for states with more coherence.

\begin{figure}[h!]
  \begin{center}
\includegraphics[width=0.9\columnwidth]{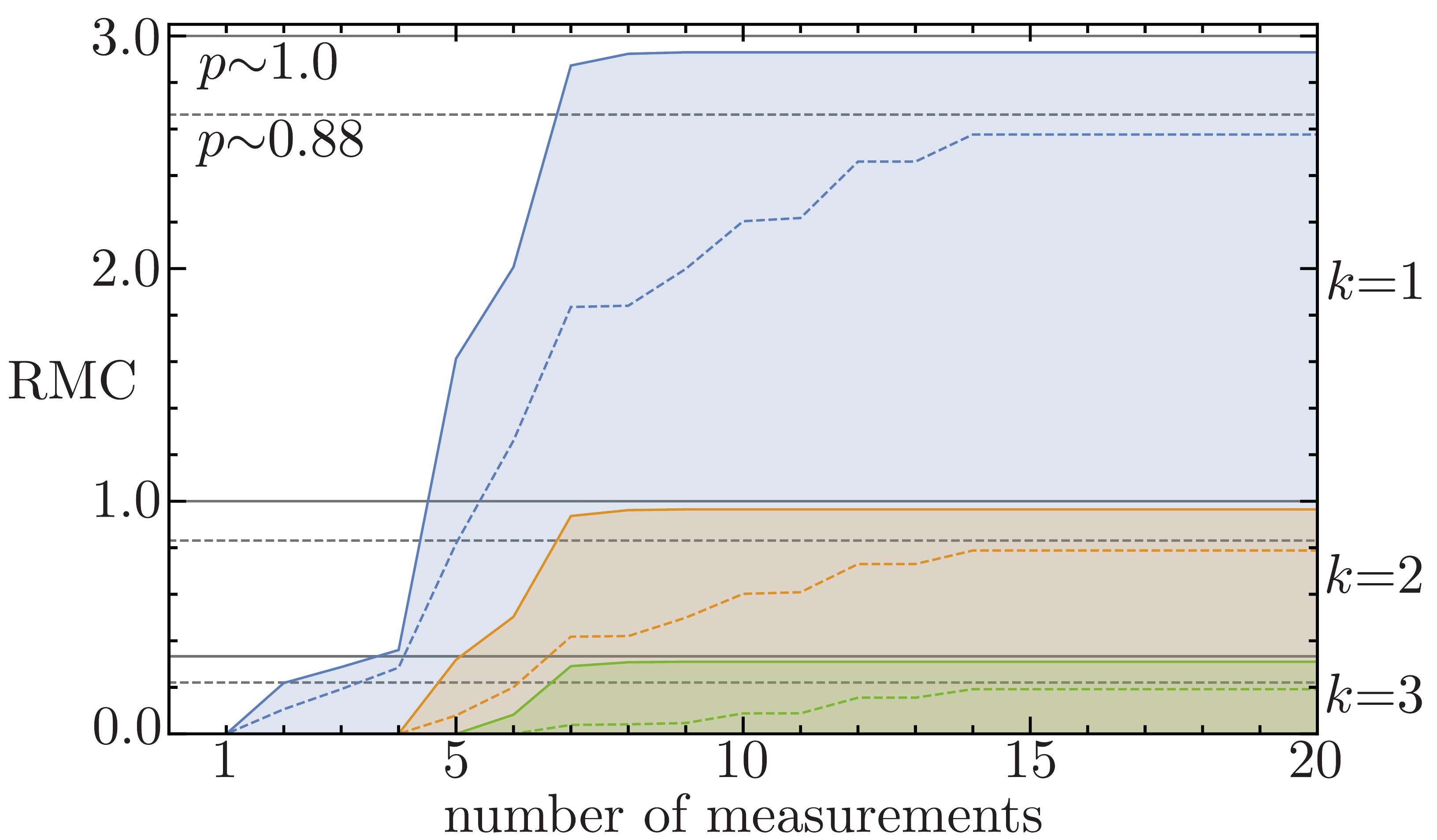}
  \end{center}
    \vspace{-2em}
\caption{\textbf{Bounding multilevel coherence from arbitrary measurements.} The blue, orange, and green solid lines correspond to the experimental lower bound on the robustness of multilevel coherence for $k=2,3,4$, respectively for a maximally coherent state $\ket{\psi_{4}^{+}}$, while the grey solid lines are the theory prediction. These bounds are obtained from the SDP in Eq.~(S16) in the Supplementary Material~\cite{SI} for an increasing number of randomly chosen projective measurements, taking into account 5$\sigma$ statistical uncertainties. The coloured dashed lines correspond to the lower bounds for the noisy maximally coherent state $\rho(0.8874\pm0.0007)$ using the same observables, with the grey dashed line being the theory prediction for this state.}
  \label{Fig:RMCObservables}
\end{figure}

\change{
We further describe how any single observable $O$ may provide a lower bound to the RMC~\cite{SI}. Consider witnesses of the form $W = \alpha \I + \beta O$, with $\alpha,\beta$ real coefficients, which then give a lower bound to the RMC via Eq.~\eqref{Eq:RobustnessDual}. Define the $k$-coherence numerical range of $O$ as the interval $\NR_{C_k}(O) = \{\Tr(O\sigma_{C_k}):\sigma\in C_k\}$ (the case $k=d-1$ was studied in \cite{johnson1981,johnsonrobinson1981}), and define its extreme points  $\lambda_{C_k}^{\min}(O)=\min \NR_{C_k}(O)$ and $\lambda_{C_k}^{\max}(O)=\max \NR_{C_k}(O)$. Notice that $\lambda^{\min}_{C_k}(O) = \min_{\CR(\ket{\psi})\leq k} \braket{ \psi | O |\psi} = \min_{I\in \mathcal{P}_k}\lambda^{\min} (P_I O P_I)$ (similarly for $\lambda^{\max}_{C_k}(O)$). Notice also that $\NR_{C_d}(O)$ is the standard numerical range of $O$, and $\lambda_{C_d}^{\min}(O)=\lambda^{\min}(O)$ (similarly for the maximal values). In general, $\NR_{C_k}(O)\subseteq \NR_{C_{k'}}(O)$ for $k\leq k'$. If $\Tr(O\rho)\in \NR_{C_k}(O)$, that is, if $\lambda_{C_k}^{\min}(O)\leq\Tr(O\rho)\leq \lambda_{C_k}^{\max}(O)$, then the expectation value of $O$ is compatible with $\rho$ being in $C_k$, and we do not gain any information on $R_{C_k}(\rho)$. If instead $\Tr(O\rho) >  \lambda_{C_k}^{\max}(O)$ or $\Tr(O\rho) < \lambda_{C_k}^{\min}(O)$, the following bound is non-trivial:
\begin{multline}
\label{eq:boundsingleobs}
R_{C_k}(\rho)\\
\geq \max\left\{0,\frac{\Tr(O\rho)-\lambda^{\max}_{C_k}(O)}{\lambda^{\max}_{C_k}(O)-\lambda^{\min}(O)},\frac{\lambda^{\min}_{C_k}(O)-\Tr(O\rho)}{\lambda^{\max}(O)-\lambda^{\min}_{C_k}(O)}\right\}.
\end{multline}
Notice that the lower bound is monotonically non-increasing with $k$.}

\subsection{Multilevel coherence as a resource for quantum-enhanced phase discrimination}
To demonstrate the operational significance of multilevel coherence we show that it is the key resource for the following task, illustrated in Fig.~\ref{Fig:PDgame}a. Suppose that a physical device can apply one of $n$ possible quantum operations $\{\Lambda_{m}\}_{m=1}^{n}$ to a quantum state $\rho$, according to the known prior probability distribution $\{p_{m}\}_{m=1}^{n}$. The output state is then subject to a single generalized measurement with elements $\{M_{m}\}_{m=1}^{n}$ satisfying $M_{m} \geq 0$ and $\sum_{m=1}^{n} M_{m} = \mathbb{I}$. Our objective is to infer the label $m$ of the quantum operation that was applied.

\begin{figure*}[ht!]
  \begin{center}
\includegraphics[width=0.9\textwidth]{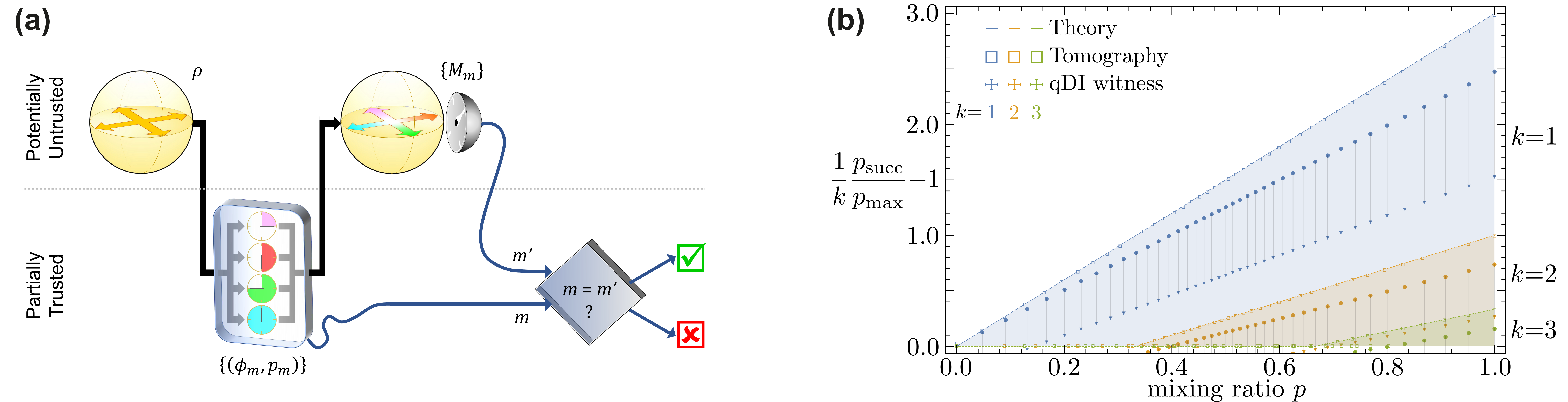}
  \end{center}
    \vspace{-2em}
\caption{\textbf{\change{Semi-device-independent} witnessing of multilevel coherence.} \textbf{(a)} A $d$-dimensional probe state $\rho$ is sent into a black box, which imprints one of $n$ phases $\{\phi_{m}\}_{m=1}^{n}$ onto the state at random according to the prior probability distribution $\{p_{m}\}_{m=1}^{n}$. To infer the index $m$ of the imprinted phase, the state is then subjected to a single generalised measurement with elements $\{M_{m}\}_{m=1}^{n}$, yielding outcome $m'$. This strategy succeeds, i.e.\ $m'=m$, with probability $p_{\rm succ}^{\Theta}(\rho)$, given by Eq.~\eqref{Eq:PSucc}, which exceeds the optimal classical success probability $p_{\max}\coloneqq\max\{p_{m}\}_{m=1}^{n}$ by a factor greater than $k$ only when $(k{+}1)$-level coherence is present in the initial state $\rho$. Since evaluating the probability of success can be done without any information about the measurement device, based only on the assumption that incoherent states are unchanged by the black box, this scheme provides a \change{semi-device-independent} method to witness and estimate the robustness of multilevel coherence in the probe.
\textbf{(b)} The experimentally measured bounds on the robustness of $(k+1)$-level coherence from the performance of noisy maximally coherent states $\rho(p)$ in the phase discrimination task $\tilde{\Theta}$, as a function of $p \in [0,1]$. Plot as in Fig.~\ref{Fig:NMCWitness}, where solid lines represent the theory predictions and open squares are the measured RMC from quantum state tomography for comparison. The filled circles (higher) correspond to the \change{semi-device-independent} witness as discussed in the text, under the assumption that the application of the phases leaves incoherent states invariant. To each filled circle corresponds an upside-down triangle (lower) which represents the conservative estimate of multilevel coherence obtained from the phase discrimination task by taking into account experimental imperfections in the implementation of the unitaries, see Supplementary Material~\cite{SI} for details. The gray lines connecting circles and corresponding triangles serve as a visual aid.
For all data, 5$\sigma$ statistical confidence regions obtained from Monte-Carlo resampling are on the order of $10^{-3}$ for $p$ and on the order of $10^{-2}$ for the RMC.}
  \label{Fig:PDgame}
\end{figure*}

We now consider a special case of these tasks, known as \emph{phase discrimination}, which is an important primitive in quantum information processing, featuring in optimal cloning, dense coding, and error correction protocols~\cite{gottesman1999fault,cerf2000asymmetric,hiroshima2001optimal,werner2001all}. Here the operations imprint a phase on the state through the transformation $\mathcal{U}_{\phi_{m}} (\rho) := U_{\phi_{m}}\rho U_{\phi_{m}}^{\dagger}$ where $U_{\phi_{m}} := \exp (-i H \phi_{m})$ is generated by the Hamiltonian $H = \sum_{j=0}^{d-1} j \ket{j}\bra{j}$. The probability of success for inferring the label $m$ in the task specified by $\Theta = \{(p_{m},\phi_{m})\}_{m=1}^{n}$ is then
\begin{equation}\label{Eq:PSucc}
p_{\rm{succ}}^{\Theta}(\rho) := \sum_{m=1}^{n}p_{m} \mbox{Tr}(\mathcal{U}_{\phi_{m}} (\rho)M_{m}).
\end{equation}
Since the Hamiltonian is diagonal in the classical basis and leaves fully incoherent states invariant, the strategy that maximizes $p_{\rm{succ}}^{\Theta}$ while at the same time only making use of incoherent states, is to guess the most likely label, that is, to take $M_m = \I\delta_{m,m_{\max}}$, succeeding with probability $p^\Theta_{\max}:= p_{m_{\max}} =\max\{p_{m}\}_{m=1}^{n}$. On the other hand, a probe state $\rho$ with non-zero coherence can outperform this strategy~\cite{Napoli2016,Piani2016}. Here, we find that genuine $(k{+}1)$-level coherence is necessary for $p_{\rm{succ}}^{\Theta}(\rho)$ to achieve a better than $k$-fold enhancement over $p^\Theta_{\rm{max}}$ in any phase discrimination task $\theta$.
\begin{theorem}
\label{Thm:PDk}
For any phase discrimination task $\Theta$ and any probe state $\rho$,
\begin{equation}\label{Eq:PSuccBoundsROCk}
\frac{p_{\rm{succ}}^{\Theta}(\rho)}{p^\Theta_{\max}} \leq k \left( 1 + R_{C_{k}} (\rho) \right).
\end{equation}
\end{theorem}
This theorem is proved in the Supplementary Material~\cite{SI}, where we also show that for the specific task $\tilde{\Theta} = \{(\frac{1}{d},\frac{2 \pi m}{d})\}_{m=1}^{d}$ of discriminating $d$ uniformly distributed phases and for a noisy maximally coherent probe, the bound in Eq.~\eqref{Eq:PSuccBoundsROCk} becomes tight. This demonstrates the key role of genuine multilevel coherence as a necessary ingredient for quantum-enhanced phase discrimination, unveiling a hierarchical resource structure which goes significantly beyond previous studies only concerned with the coarse-grained description of coherence \cite{Piani2016,Napoli2016}.

Note that this provides an operational significance to the robusteness of multilevel coherence \emph{in addition} to its operational significance in terms of resilience of noise, which in turn can be thought of also in geometric terms.

\textbf{\change{Semi-device-independent} witnessing of multilevel coherence.}
An important consequence of Eq.~\eqref{Eq:PSuccBoundsROCk} is that whenever $p_{\rm{succ}}^{\Theta}(\rho)/p^\Theta_{\max} > k$, the probe state $\rho$ must have $(k+1)$-level coherence. Consequently, the performance of an unknown state $\rho$ in any phase discrimination task $\Theta$ provides a witness of genuine multilevel quantum coherence. We remark that the success probability for an arbitrary quantum state can be evaluated without any knowledge of the devices used---neither of the one imprinting the phase, nor the final measurement. Evaluating the witness only relies on the fact that $p_{\rm{succ}}^{\Theta}(\rho) \leq p^\Theta_{\max}$ for any $\rho \in C_{1}$, which in turn relies on $\mathcal{U}_{\phi_m}(\rho) = \rho$ for any $\rho \in C_{1}$. In other words, under the condition that no information is imprinted on incoherent states, the witness can be evaluated without any additional knowledge of the used devices. We conclude that phase discrimination, as demonstrated in this paper, is a \change{\emph{semi-device-independent}} approach to measure multilevel coherence as quantified by the RMC.

Figure~\ref{Fig:PDgame} shows our experimental results for \change{semi-device-independent} witnessing of multilevel coherence using the phase discrimination task $\tilde\Theta$ for a range of noisy maximally coherent states, also taking into account experimental imperfections when it comes to the hypothesis $\mathcal{U}_{\phi_m}(\rho) = \rho$ for any $\rho \in C_{1}$ (see Supplementary Material~\cite{SI}). As any witnessing approach, this method in general only provides lower bounds on the RMC, yet in contrast to the optimal multilevel witness measured in Fig.~\ref{Fig:NMCWitness}, the present approach does not rely on any knowledge of the used measurements.

% =========================================================
% Discussion
% =========================================================
\section{Discussion}
The study of genuine multilevel coherence is pivotal, not only for fundamental questions, but also for applications ranging from transfer phenomena in many-body and complex systems to quantum technologies, including quantum metrology and quantum communication. In particular, for verifying that a quantum device is working in a nonclassical regime it is crucial to certify and quantify multilevel coherence with as few assumptions as possible. Our metrological approach satisfies these criteria by making it possible to verify the preparation of large superpositions and discriminate between them, using only the ability to apply phase transformations that leave incoherent states (approximately) invariant.
\change{The goal of the phase-discrimination task we consider is to distinguish a finite-number of phases in a single-shot scenario, and the figure of merit we adopt is the probability of success of correctly identifying the phase imprinted onto the input state. In particular, given our figure of merit, there is no notion of `closeness' of the guess to the actual phase. In contrast, in sensing applications, the task is often to measure an unknown phase with high precision~\cite{giovannetti2011}, a task we refer to as `phase estimation'. For the latter the figure of merit is the uncertainty of the estimate, and superpositions of the kind $(\ket{1}+\ket{d})/\sqrt{2}$, that is, involving eigenstates of the observable that correspond to the largest gap in eigenvalues, can be argued to be optimal~\cite{giovannetti2006}. When dealing with phase estimation, the relevant notion is that of unspeakable coherence (or asymmetry)~\cite{marvian2016quantify}, and which eigenstates are superposed is very important. On the other hand, for the kind of phase-discrimination task we consider, genuine multilevel coherence of a state like $(\ket{1}+\ket{2}+\ldots+\ket{d})/\sqrt{d}$ plays a key role. While it was already known that such a maximally coherent state provides the best performance in discriminating equally spaced phases~\cite{1976quantum}, here we find that the robustness of multilevel coherence of a generic mixed state captures its usefulness in a generic phase discrimination task. This allows us to reverse the argument, and use such usefulness to certify multilevel coherence in a semi-device-independent way.}

Our analysis of coherence rank and number, multilevel coherence witnesses, and robustness, uses and adapts notions originally studied in the context of entanglement theory~\cite{horodecki2009quantum}, and hence provides further parallels between the resource theories of quantum coherence and entanglement, whose interplay is attracting substantial interest~\cite{streltsov2016quantum}. However, a notable difference between the two that we find, emphasise, and exploit, is that multilevel coherence, unlike entanglement, can be characterised and quantified via semidefinite programming, rather than general convex optimisation~\cite{brandao2005quantifying}. This highlights multilevel coherence as a powerful, yet experimentally accessible quantum resource.

\change{Remarkably, we show that is it possible to use the notion of the comparison matrix to devise a test that faithfully detects genuine three-level coherence and above. We expect such a result to find widespread application in the study of coherence, both theoretically and experimentally. As two immediate applications, we were able to provide a full analytical classification of multilevel coherence for a qutrit, as well as to prove the existence of a ball (actually, the largest possible one, in the Hilbert-Schmidt norm) around the maximally mixed state that contains states that do not exhibit genuine multilevel coherence. This parallels the celebrated result, in entanglement theory, that there is a ball of fully separable states around the maximally mixed state of a multipartite system, and explicitly shows that generating genuine multilevel coherence is a non-trivial experimental task.}

\change{It is worth remarking that a number of our results also apply in the case of infinite-dimensional system, such as a harmonic oscillator or quantized field. Indeed, one can always consider, e.g., the quantum (multilevel) coherence exhibited by a system among a subset of states of the incoherent basis, which then provides a bound on the (multilevel) coherence in the entire Hilbert space of the system.}

Finally, our work triggers several questions to stimulate further research. These include conceptual questions regarding the exact (geometric) structure and volume of the sets $C_k$, and how sets $C_k$ and $C'_l$ defined with respect to different classical bases intersect\change{, the best further use of tools like the comparison matrix to detect and quantify multilevel coherence, or general purity-based bounds on multilevel coherence.} From a more practical point of view, a natural question is how to best choose a finite set of observables to estimate the multilevel coherence of the state of a system, for example via the SDP in the Supplementary Material~\cite{SI}. This is particularly important when one has limited access to the system under observation, as in a biological setting \cite{sarovar2010quantum,li2012witnessing,Huelga13}. Independently of the particular choice of observables, our work provides a plethora of readily applicable tools to facilitate the detection, classification, and quantitative estimation of quantum coherence phenomena in systems of potentially large complexity with minimum assumptions, paving the way towards a deeper understanding of their functional role. \change{Further theoretical investigation and experimental progress along these lines may lead to fascinating insights and advances in other branches of science where the detection and exploitation of (multilevel) quantum coherence is or can be of interest.}

\subsection*{Acknowledgments}
We thank M.~B.~Plenio, B.~Regula, V.~Scarani and A.~Streltsov for helpful discussions and T.~Vulpecula for experimental assistance. This work was supported in part by the Centres for Engineered Quantum Systems (CE110001013) and for Quantum Computation and Communication Technology (CE110001027), the Engineering and Physical Sciences Research Council (grant number EP/N002962/1), and the Templeton World Charity Foundation (TWCF 0064/AB38). We acknowledge financial support from the European Union's Horizon 2020 Research and Innovation Programme under the Marie Sk{\l}odowska-Curie Action OPERACQC (Grant Agreement No.~661338) and the ERC Starting Grant GQCOP (Grant Agreement No.~637352), and from the Foundational Questions Institute under the Physics of the Observer Programme (Grant No.~FQXi-RFP-1601).

\appendix*
%\begin{widetext}
\onecolumngrid
\medskip

\section{\ SUPPLEMENTARY MATERIAL}

% =========================================================
% Methods
% =========================================================
Here we provide detailed derivations and proofs of all results in the main text.

\subsection{Incoherent and $k$-incoherent operations.}
\change{We argue here that a fully incoherent operation is also a $k$-incoherent operation, that is, it cannot create $k+1$ coherence from states that are at most $k$-coherent.}

Consider a (fully) incoherent operation with corresponding Kraus operators $\{K_{i}\}$ \cite{baumgratz2014quantifying}. This operation is also $k$-incoherent if $\CR\left(\ket{\psi_{i}}\right) \leq k$, where $\ket{\psi_{i}} = K_{i}\ket{\psi}/\sqrt{p_{i}}$ with $p_{i} =\braket{\psi|K_{i}^{\dagger}K_{i}|\psi}$, for all pure states $\ket{\psi}$ such that $\CR(\ket{\psi}) \leq k$ and for all $i$ and $k \in \{2,3,\ldots,d\}$, which means that the Kraus operators $\{K_{i}\}$ together compose a $k$-incoherent operation. That this holds true is immediate, given that the Kraus operators $K_i$ have the form $K_i = \sum_j e^{i\phi_{i|j}} \sqrt{p_{i|j}}\ket{f_i(j)}\bra{j}$, with $f_i:\{1,\ldots,d\}\rightarrow \{1,\ldots,d\}$, $\phi_{i|j}$ a phase, and $p_{i|j}=\bra{j}K_i^\dagger K_i \ket{j}$.

\subsection{$k$-decohering operations.}
We define a \emph{$k$-decohering} map $\Lambda$ as one that destroys multilevel coherence, more precisely, such that $\Lambda(\mathscr{D(H)}) \subseteq C_{k}$. These operations generalise the notion of resource destroying maps \cite{liu2017resource} to multilevel coherence. An example of a $k$-decohering operation is the \emph{$k$-dephasing operation}
\begin{equation}\label{Eq:kDephasing}
\Delta_k(\rho) := \frac{1}{{{d-1}\choose{k-1}}} \sum_{I\in \mathcal{P}_k} P_{I} \rho P_{I},
\end{equation}
where $\mathcal{P}_{k}$ is the set of all the $k$-element subsets of $\{1, 2, \ldots, d\}$, and $P_{I} \coloneqq \sum_{i \in I} \ket{i}\bra{i}$. Since the $P_{I}$ are projectors onto $k$-dimensional subspaces, $\Delta_k(\mathscr{D(H)})\subseteq C_{k}$.
The linearity and complete positivity of $\Delta_{k}$ follow directly from the construction. Trace preservation is implied by the observation that $\Delta_{k}$ has the alternative expression
\begin{equation}\label{Eq:DeltakNiceForm}
\Delta_{k}(\rho) = \frac{k-1}{d-1} \rho + \frac{d-k}{d-1}\Delta_{1}(\rho),
\end{equation}
since $\Delta_{1}$ is clearly trace-preserving. That Eq.~\eqref{Eq:DeltakNiceForm} holds can be seen from the fact that, for all $m, n \in \{1,2,\ldots,d\}$ with $m \neq n$,
\begin{eqnarray}
\sum_{I\in \mathcal{P}_k} P_{I} \ket{m}\!\bra{m} P_{I} &=&  {{d-1}\choose{k-1}} \ket{m}\bra{m} \nonumber \\
\sum_{I\in \mathcal{P}_k} P_{I} \ket{m}\bra{n} P_{I} &=&  {{d-2}\choose{k-2}} \ket{m}\bra{n},
\end{eqnarray}
with ${{d-2}\choose{k-2}}$ the number of $I\in \mathcal{P}_k$ that contain two fixed indexes. The form of Eq.~\eqref{Eq:DeltakNiceForm} then follows directly from the definition of $\Delta_{k}$ in Eq.~\eqref{Eq:kDephasing}.

\change{
\subsection{Properties of the sets $C_{k}$.}
\change{
In this section we prove the simple sufficient condition for a state $\rho$ to be in $C_{k}$, Eq.~(7) of the main text, which also allows us to claim that every $C_k$ with $k\geq 2$ has non-zero volume within the set of all states. We also provide a proof of the characterization of the dual set $C_k^*$ presented in Theorem~2 of the main text.

\medskip
{\bf Sufficient condition for inclusion in $C_{k}$.} Using the $k$-dephasing map $\Delta_k$ of Eq.~\eqref{Eq:kDephasing}} one can characterize a non-zero-volume class of states which are in $C_k$.
Indeed, we argue that any state $\rho$ that satisfies
\begin{equation}
\label{eq:C_ksimplecondition}
\rho \geq \frac{d-k}{d-1} \Delta_1(\rho),
\end{equation}
is in $C_k$. Let us introduce the set $D_k =\{\rho \,\, | \,\, \rho\in \mathscr{D(H)},\,\rho \textrm{ satisfies~\eqref{eq:C_ksimplecondition}} \}$. Such a set is convex. It is also easy to see that, for $k\geq 2$, such a set has non-zero volume. Indeed, the maximally mixed state satisfies~\eqref{eq:C_ksimplecondition} for all $k$, with strict inequality for $k\geq 2$. This implies that, for any $k\geq 2$, any state that is an arbitrary but small enough perturbation of the maximally mixed state will still satisfy~\eqref{eq:C_ksimplecondition}. One has the following.
\begin{theorem}
\label{Thm:CkCharacterisation}
The inclusions $D_k\subseteq \Delta_k(\mathscr{D(H)}) \subseteq C_{k}$ hold for any $k=1,\ldots,d$.
\end{theorem}
\begin{proof}[Proof of Theorem~\ref{Thm:CkCharacterisation}]
For $k=1$, the inequality characterizes exactly the set of fully incoherent states $C_1$. This is because, given two normalized states $\tau$ and $\sigma$, $\tau\geq \sigma$ implies $\tau=\sigma$; in our case the implication is $\rho=\Delta_1[\rho]$.

For $k \in \{2,\ldots,d\}$, and for a state $\rho$ satisfying Eq.~\eqref{eq:C_ksimplecondition}, that is for $\rho\in D_k$, consider the operator
\begin{equation}
\sigma = \frac{d-1}{k-1} \rho - \frac{d-k}{k-1} \Delta_1(\rho),
\end{equation}
so that $\rho = \Delta_{k}(\sigma)$. We can see that $\sigma$ has unit trace by construction, and furthermore that Eq.~\eqref{eq:C_ksimplecondition} implies $\sigma \geq 0$. Thus, $\sigma \in \mathscr{D(H)}$ and $\rho\in \Delta_k(\mathscr{D(H)}))$.
\end{proof}}

\medskip
{\bf Characterising the dual sets $C_{k}^{\star}$.} Here we prove Theorem~2 of the main text, i.e.\ that $W \in C_{k}^{\star}$ if and only if $P_{I}W P_{I} \geq 0$ for all $I \in \mathcal{P}_{k}$.

\begin{proof}[Proof of Theorem~2]
Formally, the set $C^*_k$ of $(k{+}1)$-level coherence witnesses is obtained as the dual of the set $C_k$ and given by
\begin{equation}\label{Eq:CkStar}
C^*_k = \{W : W=W^\dagger,\, \Tr(W\sigma) \geq 0,\,\,\forall\,\sigma\in C_k\}.
\end{equation}
This definition, together with the convexity of $C_k$ implies that it is sufficient to see that
\begin{equation}
\braket{\psi |W| \psi} \geq 0 \quad \forall \; \ket{\psi}\textrm{ such that }\CR (\ket{\psi}) \leq k,
\end{equation}
if and only if $P_{I}W P_{I} \geq 0$ for all $I \in \mathcal{P}_{k}$. This is immediate since, on the one hand, for any given $I \in \mathcal{P}_{k}$ the action of projecting with $P_{I}$ on an arbitrary $\ket{\psi} \in \mathscr{D(H)}$ is either to return the null vector or a pure state (up to normalisation) with coherence rank not exceeding $k$. On the other hand, for any $\ket{\psi}$ such that $\CR (\ket{\psi}) \leq k$, one can always find an $I \in \mathcal{P}_{k}$ such that $P_{I} \ket{\psi} = \ket{\psi}$.
\end{proof}

\change{\subsection{Witness of multilevel coherence for a pure state}
\label{sec:witnesspurestateSI}
In the main text we already argue that the witness
\begin{equation}
\label{eq:maxcohwitness_SI}
W_k(\psi) = \I - \frac{1}{\sum_{i=1}^k |c_i^{\downarrow}|^2}\proj{\psi} ,
\end{equation}
where $c_{i}^{\downarrow}$ are the coefficients $c_{i}$ rearranged in non-increasing modulus order, detects the $(k+1)$-multilevel coherence of the state $\ket{\psi}$, if present, by means of a negative expectation value $\braket{\psi|W_k(\psi)|\psi}$. Here we provide details of the proof that $W_k(\psi)$ is a proper multilevel-cohrence witness, that is, that $\mbox{Tr}(W \sigma) \geq 0$ for all $\sigma\in C_{k}$. Given that $C_k$ is a convex set whose extreme points are pure states with coherence rank less or equal to $k$, and given that the trace functional is linear in its argument, it is enough to check that $\braket{\phi|W_k(\psi)|\phi}\geq 0$ for all pure states $\ket{\phi}$ with coherence rank less of equal to $k$. Given the structure of $W_k(\psi)$, this is equivalent to proving that
$\max_{\CR(\phi)\leq k} |\braket{\phi|\psi}|^2 \leq  \sum_{i=1}^k |c_i^{\downarrow}|^2$ (actually, with equality). This is readily proven as follows. Let $\ket{\phi}=\sum_{i\in I_\phi}\phi_i\ket{i}$ be the decomposition of
a generic $\ket{\phi}$ with $\CR(\phi)\leq k$, with $I_\phi\subseteq \{1,2,\dots,d\}$ the $\phi$-dependent subset of at most $k$ coefficients $\phi_i$ that do not vanish. Then
\begin{equation}
\begin{split}
|\braket{\phi|\psi}|
&= \left|\sum_{i\in I_\phi} \phi_i^*c_{i}^{\downarrow}\right|\\
&\leq \sqrt{\sum_{i\in I_\phi} |\phi_i|^2}\sqrt{\sum_{i\in I_\phi} |c_{i}^{\downarrow}|^2} \\
&\leq \max_{I:|I|\leq k} \sqrt{\sum_{i\in I} |c_{i}^{\downarrow}|^2}\\
& = \sqrt{\sum_{i=1}^k |c_{i}^{\downarrow}|^2},
\end{split}
\end{equation}
where: the first inequality is due to the Cauchy-Schwarz inequality; the second inequality comes from the fact that $\ket{\phi}$ is a normalized state and that we optimize over any index set $I$ of cardinality $|I|\leq k$, rather than being restricted to $I_\phi$; the last equality comes from the ordering of the coefficients $c_i^\downarrow$.}

\change{
\subsection{Eigenvalues of multilevel-coherence witnesses.}
As mentioned in Observation~1 in the main text, the number of negative eigenvalues of a multilevel-coherence witness is severely constrained: $W \in C_{k}^{\star}$ only if it has at most $d-k$ negative eigenvalues (counting multiplicity). Let also $\lambda^{\min}(X)$ ($\lambda^{\max}(X)$) denote the smallest (largest) eigenvalues of $X=X^\dagger$. If $W\in C^\star_{k}$, then its most negative eigenvalue, $\lambda^{\min}(W)$, satisfies
\begin{equation}
\lambda^{\min}(W)\geq -\frac{d-k}{k}\lambda^{\max}(W).
\label{eq:maxnegW}
\end{equation}}

\begin{proof}[Proof of Observation~1]
\change{Let $\lambda^\downarrow_i(X)$ be the eigenvalues of $X=X^\dagger$, ordered in monotonically decreasing order, so that $\lambda^{\min}(X)=\lambda^\downarrow_d(X)$ and $\lambda^{\max}(X)=\lambda^\downarrow_1(X)$.}

\change{The interlacing theorem~\cite{bhatiamatrixanalysis} states that, if $A$ is a $d\times d$ Hermitian matrix, and $B = PAP$, with $P$ a $k$-dimensional projection, then one has $\lambda^\downarrow_j(A)\geq\lambda^\downarrow_j(B)\geq\lambda^\downarrow_{j+d-k}(A)$, for $j=1,\ldots,k$.	In our case $A=W$, and $B=P_I W P_I$, for $I\in \mathcal{P}_k$. Since $W$ is assumed to be in $C_k^\star$, $B\geq 0$, that is, $\lambda^\downarrow_j(B)\geq 0$ for all $j=1,\ldots,d$. Hence, also the largest $k$ eigenvalues of $A$ must be non-negative, that is, $A$ can at most have $d-k$ negative eigenvalues.}
	
\change{For the proof of Eq.~\eqref{eq:maxnegW}, let $\ket{w^{\min}}$ be the eigenstate corresponding to the lowest eigenvalue $\lambda^{\min}(W)$ of $W$. One has
\begin{equation}
\label{eq:boundW}
\begin{split}
0&\leq \min_{\CR(\ket{\phi})\leq k} \braket{\phi | W | \phi} \\
 &\leq \min_{\CR(\ket{\phi})\leq k} \left(\lambda^{\min}(W)|\braket{\phi | w^{\min}}|^2+\lambda^{\max}(W)(1-|\braket{\phi | w^{\min}}|^2)\right)\\
 &=  \lambda^{\min}(W)\max_{\CR(\ket{\phi})\leq k}|\braket{\phi | w^{\min}}|^2+\lambda^{\max}(W)(1-\max_{\CR(\ket{\phi})\leq k}|\braket{\phi | w^{\min}}|^2)\\
 &\leq \lambda^{\min}(W) \frac{k}{d}+\lambda^{\max}(W)\frac{d-k}{d} ,
\end{split}
\end{equation}
where the first inequality is due to $W\in C_k^\star$, and the second inequality follows from the definition of largest and smallest eigenvalue of $W$ and the normalization of $\ket{w^{\min}}$. In the last inequality, we have used the fact that, for any fixed $\ket{\psi}=\sum_{i=1}^d c_i \ket{i}$, it holds $\max_{\CR(\ket{\phi})\leq k}|\braket{\phi | \psi}|^2 = \sum_{i=1}^k |c^\downarrow_i|^2$ (see Section~\ref{sec:witnesspurestateSI}), where $\{c^\downarrow_i\}$ are the coefficients of $\ket{\psi}$ rearranged in non-increasing order, with respect to their modulus.  In turn, because of normalization of $\ket{\psi}$, and because of the ordering of the coefficients $c^\downarrow_i$, it holds $\sum_{i=1}^k |c^\downarrow_i|^2\geq k/d$. Eq.~\eqref{eq:maxnegW} is obtained by simple rearrangement.}
\end{proof}
\change{It is worth remarking that the $k+1$-level-coherence witness $W_{k}(\psi^+_d)$ in Section 1D of the main text saturates the bound on Eq.~\eqref{eq:maxnegW}. Also, any optimal multilevel-coherence witness (that is, a witness whose expectation value gives the RMC), has necessarily largest eigenvalue equal to $1$, in order to satisfy tightly the constraint $W\leq \I$ in Eq.~(11); thus, any \emph{optimal} $(k+1)$-level-coherence witness $W_{k}$ actually satisfies $\lambda^{\min}(W_{k})\geq -\frac{d-k}{k}$.}

\change{
\subsection{Analytical criterion for genuine multilevel coherence}
We now present the proof of Theorem~1 in the main text and deduce from it the condition in Eq.~(8), which identifies the largest Euclidean ball centered around the maximally mixed state and entirely contained inside $C_2$.
We start by remarking that by Carath\'eodory's theorem there exists a decomposition of any $d\times d$ density matrix $\rho$ into at most $d^2$ pure states with coherence rank at most $\CN(\rho)$. This entails---among other things---that the function $\CN$ is lower semicontinuous, i.e.\ that if a sequence of density matrices $\rho_n$ satisfies $\CN(\rho_n)\leq k$ and $\lim_{n\to \infty} \rho_n = \rho$ then also $\CN(\rho)\leq k$.

We remind the reader that a $d\times d$ matrix $A$ is said to be \emph{diagonally dominant} if~\cite[Definition~6.1.9]{HJ1}
\begin{equation}
|A_{ii}| \geq \sum_{j\neq i} |A_{ij}| \qquad \forall\ i=1,\ldots, d\, ,
\label{diagonal dominance}
\end{equation}
and \emph{strictly diagonal dominant} if Eq.~\eqref{diagonal dominance} holds with strict inequality. Diagonally dominant matrices enjoy many useful properties, some of which are as follows.

\begin{lemma} \emph{\cite[Corollary~5.6.17]{HJ1} and \cite[Problem~7~p.~40]{HJ2}.} \label{lemma dd positive}
Strictly diagonally dominant matrices are invertible. A Hermitian matrix with non-negative diagonal that is also diagonally dominant is positive semidefinite.
\end{lemma}

\begin{lemma} \emph{\cite[Theorem~1]{Varga76} or \cite[Theorem~2.5.3 and Exercise~p.~124]{HJ2}.} \label{lemma congruence dd}
Given a Hermitian matrix $A$, there exists a diagonal matrix $D>0$ such that $D A D$ is strictly diagonally dominant if and only if $M(A) > 0$.
\end{lemma}

As explained in \cite{Varga76} (see the discussion after Eq.~(3.8) there) the above lemma can be deduced from the equivalence of conditions (ii) and (iii) in \cite[Theorem~1]{Varga76}. In fact, it is not difficult to see that the `generalised column diagonal dominance' condition (iii) is equivalent to the existence of $D>0$ such that $DAD$ is strictly diagonally dominant, while condition (ii) translates directly to $M(A)>0$ because $M(A)$ is Hermitian and hence its eigenvalues are automatically real.

From Lemmas \ref{lemma dd positive} and \ref{lemma congruence dd} one can immediately deduce a quick criterion for positivity, also widely known.

\begin{corollary} \label{cor positivity comparison}
Given a $d\times d$ Hermitian matrix $A$ with non-negative diagonal entries, if $M(A)\geq 0$ then also $A\geq 0$.
\end{corollary}

\begin{proof}
Let $A$ satisfy the hypotheses. For all $\epsilon>0$ set $A_\epsilon \coloneqq A + \epsilon \I$. Then $M(A_\epsilon)=M(A)+ \epsilon \I >0$, and by Lemma~\ref{lemma congruence dd} there exists $D_\epsilon >0$ diagonal such that $D_\epsilon A_\epsilon D_\epsilon$ is strictly diagonally dominant and hence positive definite by Lemma~\ref{lemma dd positive}. Since $D_\epsilon$ is invertible, we deduce that also $A_\epsilon$ is positive definite, i.e.\ $A_\epsilon > 0$. Upon taking the limit $\epsilon\to 0^+$ we obtain that $A\geq 0$.
\end{proof}

Another notable corollary is a Hermitian version of a well-known theorem by Camion and Hoffman \cite{Camion1966} (see also \cite[Theorem~2.5.14]{HJ2}). For a fixed dimension $d$, let us construct the following set of $d\times d$ matrices:
\begin{equation}
\Omega \coloneqq \left\{ \omega=\omega^\dag:\ |\omega_{ij}|=1=\omega_{ii}\quad \forall\ i,j \right\} .
\end{equation}
Observe that $\Omega \subseteq C_2^*$, as one can verify directly by employing Theorem~2 from the main text.
In what follows, we will find it convenient to work with \emph{Hadamard products}~\cite[Definition~7.5.1]{HJ1}. We remind the reader that the Hadamard product of two $d\times d$ matrices $A,B$ is another $d\times d$ matrix $A\circ B$ whose entries are simply defined by $(A\circ B)_{ij}\coloneqq A_{ij}B_{ij}$. For a matrix $A$, we set
\begin{equation}
\Omega \circ A \coloneqq \left\{ \omega \circ A:\ \omega \in \Omega \right\} .
\end{equation}
We now recall the following result. Although it is part of the folklore on the subject, we include a proof for completeness.

\begin{corollary} \label{cor equimodular}
A $d\times d$ Hermitian matrix $A$ with non-negative diagonal entries is such that $\Omega\circ A$ comprises only positive semidefinite matrices if and only if $M(A)\geq 0$.
\end{corollary}

\begin{proof}
Since $M(A)\in \Omega \circ A$, clearly $M(A)\geq 0$ is necessary to ensure that $\Omega \circ A$ is composed only of positive semidefinite matrices. To show the converse, observe that for all $\omega\in \Omega$ we have that $M(\omega \circ A)=M(A)$. By virtue of Corollary~\ref{cor positivity comparison}, it is then clear that $M(A)\geq 0$ is also a sufficient condition.
\end{proof}

We are now ready to prove Theorem~1 from the main text, which states that the condition $M(\rho)\geq 0$ is necessary and sufficient in order for $\rho$ to have coherence number at most $2$.

\begin{proof}[Proof of Theorem~1]
We start by showing that $M(\rho)\geq 0$ is sufficient to ensure that $\CN (\rho)\leq 2$. For $0<\epsilon\leq 1$, set $\rho_\epsilon \coloneqq (1-\epsilon)\rho + \epsilon \frac{\I}{d}$. As in the proof of Corollary~\ref{cor positivity comparison}, we have $M(\rho_\epsilon)>0$. Hence, by Lemma~\ref{lemma congruence dd} there exists $D_\epsilon>0$ such that $D_\epsilon \rho_\epsilon D_\epsilon$ is strictly diagonally dominant. Since it is easy to see that diagonally dominant density matrices have coherence number at most $2$ \cite{Johnston18}, and moreover the coherence number is invariant under congruence by invertible diagonal matrices, we obtain that
\begin{equation*}
\CN(\rho_\epsilon) = \CN \left(\frac{D_\epsilon \rho_\epsilon D_\epsilon}{\Tr(D_\epsilon \rho_\epsilon D_\epsilon)}\right)\leq 2\, .
\end{equation*}
Taking the limit $\epsilon\to 0^+$ and using the fact that $\CN$ is lower semicontinuous we see that in fact also
\begin{equation*}
\CN (\rho) = \CN \left( \lim_{\epsilon\to 0^+} \rho_\epsilon\right) \leq \lim_{\epsilon\to 0^+} \CN (\rho_\epsilon) \leq 2\, .
\end{equation*}

We now show the converse, i.e.\ that every density matrix $\rho$ such that $\CN (\rho)\leq 2$ satisfies $M(\rho)\geq 0$. To this end, we prove that $\CN (\rho)\leq 2$ implies that $\Omega \circ \rho$ is composed only of positive semidefinite matrices, and then the claim will follow from Corollary~\ref{cor equimodular}. If $\rho = \sum_{\alpha} p_{\alpha} \psi_{\alpha}$ is a decomposition of $\rho$ such that every $\ket{\psi_{\alpha}}$ has coherence rank at most $2$, for all $\omega\in \Omega$ we have
\begin{equation}
\omega \circ \rho = \sum_{\alpha} p_{\alpha}\, \omega \circ \psi_{\alpha} \geq 0\, ,
\end{equation}
where the last inequality follows because if $A$ is nonzero only on a $2\times 2$ subspace then $\omega \circ A = U A U^\dag$ for some diagonal unitary $U$.
\end{proof}

From the above result we can deduce an easy criterion that allows a complete classification of qutrit states based on their coherence number. Notice that it is trivial to check whether $\CN(\rho)=1$, since in the latter case $\rho$ is simply diagonal, so the non-trivial part of the classification is in distinguishing between the case $\CN(\rho)\leq 2$ and the case $\CN(\rho)=3$.

\begin{corollary}
Let $\rho$ be a qutrit state. If $\det M(\rho)\geq 0$ then $\CN(\rho)\leq 2$, otherwise $\CN(\rho)=3$.
\end{corollary}

\begin{proof}
Since all $2\times 2$ principal minors of $M(\rho)$ have the same determinant as the corresponding minor of $\rho$, hence non-negative, Sylverster's criterion ensures that $M(\rho)\geq 0$ if and only if $\det M(\rho)\geq 0$. Theorem~1 then ensures that this latter condition is necessary and sufficient in order for $\rho$ to have coherence number at most $2$.
\end{proof}

In what follows, we employ Theorem~1 to deduce the condition Eq.~(8) of the main text.

\begin{corollary} \label{cor largest ball}
For all $d\times d$ density matrices $\rho$,
\begin{equation}
\Tr \left(\rho^2\right)\leq \frac{1}{d-1}\quad \Longrightarrow\quad \CN(\rho)\leq 2\, .
\label{largest ball}
\end{equation}
\end{corollary}

\begin{proof}
We show that if $\CN(\rho)>2$ then necessarily $\Tr\left(\rho^2\right)>\frac{1}{d-1}$. By Theorem~1, $\CN(\rho)>2$ implies that the comparison matrix $M(\rho)$ is not positive semidefinite.
Observe that $\Tr (M(\rho)) = \Tr (\rho) =1$, and moreover $\Tr \left( M(\rho)^2\right) = \Tr (\rho^2)$. The first equation tells us that since one of the eigenvalues of $M(\rho)$ is negative, the sum of the remaining $d-1$ must be larger than $1$. For fixed $\sum_{i=1}^{d-1} \lambda_i = c > 1$, the quantity $\sum_{i=1}^{d-1}\lambda_i^2$ is well-known to be minimised when $\lambda_i\equiv \frac{c}{d-1}$ for all $i$. Hence, in this case we would obtain $\Tr\left(\rho^2\right) = \Tr \left(M(\rho)^2\right)\geq \frac{c^2}{d-1}>\frac{1}{d-1}$. This concludes the proof.
\end{proof}

\begin{remark}
Observe that the condition on the l.h.s. of \eqref{largest ball} is equivalent to the condition
\[
\left\|\, \rho - \frac{\I}{d}\,\right\|_2 \leq \frac{1}{\sqrt{d(d-1)}},
\]
that is, to the condition that $\rho$ is close enough to the maximally mixed state $\I/d$ in the Hilbert-Schmidt norm $\|X\|_2=\sqrt{\Tr(X^\dagger X)}$.
Indeed, the above result provides the size of the \emph{maximal} Euclidean ball centered around the maximally mixed state that lies inside the set $C_2$ of $2$-coherent states. There can not be any larger such ball, as the one we constructed already touches the boundary of the set of $d\times d$ density matrices. This is because it contains some rank-deficient states, e.g.\ all normalised projectors onto $(d-1)$-dimensional subspaces.
\end{remark}}

\subsection{Robustness of $k{+}1$ coherence.} It is possible to write $R_{C_{k}}(\rho)$ as the solution of the following SDP:
\begin{equation}
\label{Eq:RobustnessPrimal}
\begin{array}{lll}
R_{{C}_k} (\rho) \,\,\, = \,\,\, & \min  \qquad & \Tr(\sum_{I\in\mathcal{P}_k}\tilde{\sigma}_{I}) -1 \\[6pt]
& \textup{s.t.} & \tilde{\sigma}_{I} \geq 0 \quad\forall I\in \mathcal{P}_k\\[6pt]
&  & P_{I}\tilde{\sigma}_{I}P_I = \tilde{\sigma}_I \quad\forall I\in \mathcal{P}_k\\[6pt]
&  & \sum_{I\in\mathcal{P}_k} \tilde{\sigma}_{I} \geq \rho \, . \\
\end{array}
\end{equation}
The dual SDP is given by Eq.~(11) in the main text. We now show that Eq.~\eqref{Eq:RobustnessPrimal} holds. First, one may rewrite $R_{C_{k}}(\rho)$ as
\begin{equation}\label{Eq:RMCAlternative}
R_{C_{k}}(\rho) = \inf_{\sigma \in C_{k}} \left\lbrace s \geq 0 : \rho \leq (1+s)\sigma \right\rbrace.
\end{equation}
One then arrives to Eq.~\eqref{Eq:RobustnessPrimal} by using the defying property of $C_k$, that is, that $\sigma \in C_{k}$ can be written as the convex combination of pure states with coherence rank at most $k$. Thus, for any $\tilde{\sigma} \geq 0$, we have that $\frac{\tilde{\sigma}}{\mbox{Tr}(\tilde{\sigma})} \in C_{k}$ if and only if $\tilde{\sigma} = \sum_{I \in \mathcal{P}_{k}} \tilde{\sigma}_{I}$, such that for all $I \in \mathcal{P}_{k}$ it holds that $P_{I}\tilde{\sigma}_{I}P_{I} = \tilde{\sigma}_{I}$ and $\tilde{\sigma}_{I} \geq 0$.

We note that strong duality holds trivially since $\mathbb{I}\in C_{k}^{\star}$.

{\bf Robustness of multilevel coherence of the noisy maximally coherent states.}
It is simple to check that for the witness $W_{k}(\psi_{d}^{+})$ in Eq.~(9) in the main text we have $\Tr(W_k(\psi_d^+) \rho(p)) = \frac{1}{k}[(k-1)-p(d-1)]$ and $W_{k}(\psi_{d}^{+}) \leq \mathbb{I}$.
One then concludes that $R_{C_{k}}(\rho (p)) \geq \max \left\lbrace 0 ,  - \Tr(W_k(\psi_d^+)\rho (p))  \right\rbrace = \max \left\lbrace 0 , \frac{1}{k}[p(d-1)-(k-1)] \right\rbrace$. On the other hand, it can be seen that $\rho(p) \leq (1+ s(\rho(p))) \Delta_{k}(\ket{\psi_{d}^{+}}\bra{\psi_{d}^{+}})$ for $s(\rho(p)) = \max \left\lbrace 0 , \frac{1}{k}[p(d-1)-(k-1)] \right\rbrace$. Then, from Eq.~\eqref{Eq:RMCAlternative} we see that $R_{C_{k}}(\rho (p)) \leq s(\rho (p))$ and can hence conclude that
\begin{equation}\label{Eq:RCkNoisyMaximallyCoherent}
R_{C_{k}}(\rho (p)) = \max \left\lbrace 0 , \frac{1}{k}[p(d-1)-(k-1)] \right\rbrace.
\end{equation}

\change{
\subsection{Bounding RMC from one measurement via witnesses.}
For a given observable $O$ and corresponding expectation value $\Tr(O\rho)$, we consider witnesses of the form $W = \alpha \I + \beta O$, with $\alpha,\beta$ real coefficients. The idea is that, since we focus on a subset of all possible witnesses, we will bound the robustness of multilevel coherence exploiting Eq.~(11), via
\begin{equation}
\begin{array}{lll}
R_{{C}_k} (\rho) \,\,\, \geq \,\,\, & \max  \qquad & - \Tr((\alpha \I + \beta O)\rho)) \\[6pt]
& \textup{s.t.} & \Tr((\alpha \I + \beta O)\sigma) \geq 0 \quad \forall \sigma\in C_k\\[6pt]
&  & \alpha \I + \beta O \leq \I  \\[6pt]
&  & \alpha,\beta\in \mathbb{R}\, .
\end{array}
\end{equation}
This bound simplifies to
\begin{equation}
\label{eq:singleobservablebound}
\begin{array}{lll}
R_{{C}_k} (\rho) \,\,\, \geq \,\,\, & \max  \qquad & - (\alpha + \beta\Tr(\rho O)) \\[6pt]
& \textup{s.t.} & \alpha+\beta\braket{\phi | O | \phi} \geq 0 \quad \forall\ket{\phi}\quad\text{s.t.}\quad \CR(\ket{\phi})\leq k\\[6pt]
&  & \alpha + \beta \braket{\psi | O | \psi} \leq 1 \quad\forall \ket{\psi} \\[6pt]
&  & \alpha,\beta\in \mathbb{R}\, .
\end{array}
\end{equation}
For the convenience of the reader, we restate some of the definitions given in the main text. We define the $k$-coherence numerical range of $O$ as the interval $\NR_{C_k}(O) = \{\Tr(O\sigma_{C_k}):\sigma\in C_k\}$, and define its extreme points  $\lambda_{C_k}^{\min}(O)=\min \NR_{C_k}(O)$ and $\lambda_{C_k}^{\max}(O)=\max \NR_{C_k}(O)$. Notice that $\lambda^{\min}_{C_k}(O) = \min_{\CR(\ket{\psi})\leq k} \braket{ \psi | O |\psi} = \min_{I\in \mathcal{P}_k}\lambda^{\min} (P_I O P_I)$ (similarly for $\lambda^{\max}_{C_k}(O)$). Notice also that $\NR_{C_d}(O)$ is the standard numerical range of $O$, and $\lambda_{C_d}^{\min}(O)=\lambda^{\min}(O)$, where $\lambda^{\min}(X)$ is the smallest eigenvalues of $X=X^\dagger$ (similarly for the maximal values).
It is convenient to split the optimization \eqref{eq:singleobservablebound} into the two cases $\beta<0$ and $\beta\geq 0$.
If we optimize over $\beta<0$, the bound assumes the form
\begin{equation}
\begin{array}{lll}
R_{{C}_k} (\rho) \,\,\, \geq \,\,\, & \sup  \qquad & - (\alpha + \beta\Tr(\rho O)) \\[6pt]
& \textup{s.t.} & \alpha+\beta \lambda_{C_k}^{\max}(O) \geq 0\\[6pt]
&  & \alpha + \beta \lambda^{\min}(O) \leq 1 \\[6pt]
&  & \alpha, \beta\in \mathbb{R}, \quad \beta<0,
\end{array}
\end{equation}
while, optimizing over $\beta \geq 0$, we have
\begin{equation}
\begin{array}{lll}
R_{{C}_k} (\rho) \,\,\, \geq \,\,\, & \max  \qquad & - (\alpha + \beta\Tr(\rho O)) \\[6pt]
& \textup{s.t.} & \alpha+\beta \lambda_{C_k}^{\min}(O) \geq 0\\[6pt]
&  & \alpha + \beta \lambda^{\max}(O) \leq 1 \\[6pt]
&  & \alpha, \beta\in \mathbb{R}, \quad \beta \geq 0.
\end{array}
\end{equation}
These two separate optimizations are easily handled. We present the details for the one for $\beta<0$; the one for $\beta\geq0$ is handled similarly. For $\beta<0$, we want to take $\alpha$ as small as possible; at the same time $\alpha$ must satisfy
\begin{equation*}
- \beta \lambda_{C_k}^{\max}(O) \leq \alpha \leq 1-\beta \lambda^{\min}(O).
\end{equation*}
Thus, the optimization is feasible if
\begin{equation*}
- \beta \lambda_{C_k}^{\max}(O) \leq 1-\beta \lambda^{\min}(O) \Leftrightarrow  \beta \geq -\frac{1}{\lambda_{C_k}^{\max}(O)-\lambda^{\min}(O)}.
\end{equation*}
Notice that $\lambda_{C_k}^{\max}(O) \geq (k/d)\lambda^{\max}(O)+ ((d-k)/d)\lambda^{\min}(O)$ (proven along the lines of Eq.~\eqref{eq:boundW}), so that the denominator in the last expression is strictly positive for all $k\geq 1$ as long as $O$ is not fully degenerate (in the latter case measuring $O$ clearly can not provide any information).
If $\beta<0$ is in the feasible region, we want to take $\alpha=- \beta \lambda_{C_k}^{\max}(O)$, so that the target value is $-\beta(\Tr(O\rho) - \lambda_{C_k}^{\max}(O))$. If $\Tr(O\rho) - \lambda_{C_k}^{\max}(O) \geq 0$, the largest value is obtained by choosing $\beta$ as negative as possible, that is $\beta = -1/(\lambda_{C_k}^{\max}(O)-\lambda^{\min}(O))$, so that the optimal value is $(\Tr(O\rho)-\lambda^{\max}_{C_k}(O))/(\lambda^{\max}_{C_k}(O)-\lambda^{\min}(O))$; otherwise, if $\Tr(O\rho) - \lambda_{C_k}^{\max}(O) \geq 0$, we take $\beta\rightarrow 0^-$, with optimal value $0$.

Thus, considering also the case $\beta \geq 0$, we arrive at the bound of Eq.~(12) of the main text:
\begin{equation}
R_{C_k}(\rho)
\geq \max\left\{0,\frac{\Tr(O\rho)-\lambda^{\max}_{C_k}(O)}{\lambda^{\max}_{C_k}(O)-\lambda^{\min}(O)},\frac{\lambda^{\min}_{C_k}(O)-\Tr(O\rho)}{\lambda^{\max}(O)-\lambda^{\min}_{C_k}(O)}\right\},
\end{equation}
which is non-trivial when $\Tr(O\rho)\neq \NR_{C_k}(O)$, that is, when the expectation value of $O$ is not compatible with $\rho$ being in $C_k$. Notice that a similar approach is possible in quantifying other resources, like entanglement. One feature that makes multilevel coherence special is that the multilevel-coherence numerical range can be explicitly calculated.}

\subsection{Bounding RMC from arbitrary measurements.}
If one has access to the expectation values $o_{i} = \mbox{Tr}(O_{i}\rho)$ of a set of $n$ observables $\{O_{i}\}_{i=1}^{n}$ measured on an experimentally prepared state $\rho$, one can lower bound the RMC using the SDP:
\begin{equation}\label{Eq:RobustnessObservablesPhys}
\hspace*{-.2cm}\begin{array}{lll}
R_{{C}_k} (\rho) \geq & \min & \mbox{Tr}(\sum_{I\in\mathcal{P}_k}\tilde{\sigma}_{I})-1 \\[6pt]
& \textup{s.t.} & \tilde{\sigma}_{I} \geq 0 \quad \forall I \in \mathcal{P}_{k} \\[6pt]
&  & P_{I}\tilde{\sigma}_{I}P_I = \tilde{\sigma}_I \quad \forall I \in \mathcal{P}_{k} \\[6pt]
&  & \sum_{I\in\mathcal{P}_k} \tilde{\sigma}_{I} \geq \tau \\[6pt]
& & \tau \geq 0 \\[6pt]
& & \mbox{Tr}(\tau) = 1 \\[6pt]
& & o_{i} - e_{i}^- \leq \mbox{Tr}(O_{i}\tau ) \leq  o_{i} + e_{i}^+ \quad \forall i\, , \\[6pt]
\end{array}
\end{equation}
where we allow for the lower and upper experimental uncertainties $e_{i}^-$ and $e_{i}^+$, respectively. Here we look for an optimal $\tau \in \mathcal{D(H)}$ that satisfies $\sum_{I\in\mathcal{P}_k} \tilde{\sigma}_{I} \geq \tau$ while being consistent with the results of the expectation values, that is $\mbox{Tr}( O_{i}\tau )=\mbox{Tr}(O_{i}\rho)$, within experimental uncertainties.

Note that the SDP in Eq.~\eqref{Eq:RobustnessObservablesPhys} only requires the optimization to reproduce the measured expectation values to within the supplied error bounds. This leads to a trade-off, where smaller error bounds to the SDP lead to closer convergence to the actual value of multilevel coherence, while larger error bounds improve the stability of the estimation against statistical fluctuations. In the experiment presented in Fig.~5 of the main text, we have chosen conservative $5\sigma$ error bounds, leading to a deviation about $5\%$ between the lower bound and the tomographically estimated RMC. This could be improved by incorporating maximum-likelihood or Bayesian estimation techniques as in the case of quantum tomography.

\subsection{Phase discrimination.} Consider the optimal $\sigma^{\star} \in C_{k}$ satisfying the optimisation in Eq.~\eqref{Eq:RMCAlternative} for a given state $\rho$, i.e.\ such that $\rho \leq (1 + R_{C_{k}}(\rho))\sigma^{\star}$. Following from the linearity of the success probability in Eq.~(13) in the main text, we see that
\begin{equation}
p_{\rm{succ}}^{\Theta}(\rho) \leq (1 + R_{C_{k}}(\rho))p_{\rm{succ}}^{\Theta}(\sigma^{\star}).
\end{equation}
Now, by setting $k=1$, one may also consider the optimal $\delta^{\star} \in C_{1}$ such that $\sigma^{\star} \leq (1 + R_{C_{1}}(\sigma^{\star}))\delta^{\star}$, which means
\begin{equation}
p_{\rm{succ}}^{\Theta}(\sigma^{\star}) \leq (1 + R_{C_{1}}(\sigma^{\star}))p_{\rm{succ}}^{\Theta}(\delta^{\star}).
\end{equation}
Overall then, we have
\begin{equation}
p_{\rm{succ}}^{\Theta}(\rho) \leq (1 + R_{C_{k}}(\rho))(1 + R_{C_{1}}(\sigma^{\star}))p_{\rm{succ}}^{\Theta}(\delta^{\star}).
\end{equation}
Since $\sigma^{\star} \in C_{k}$, we find that $R_{C_{1}}(\sigma^{\star})\leq k-1$. Furthermore, we have already seen that $p_{\rm{succ}}^{\Theta}(\delta^{\star})\leq p_{\max}$. Hence, we arrive at
\begin{equation}
p_{\rm{succ}}^{\Theta}(\rho) \leq (1 + R_{C_{k}}(\rho))k p_{\max}
\end{equation}
which can be rearranged to Eq.~(14) in the main text.

When one considers the phase discrimination task $\tilde{\Theta}$ with a probe prepared in the noisy maximally coherent state $\rho(p)$ and optimised generalised measurements $\tilde{M}_{m} = \mathcal{U}_{\phi_{m}}(\ket{\psi_{d}^{+}}\bra{\psi_{d}^{+}})$, the success probability is~\cite{Napoli2016,Piani2016}
\begin{equation}
p_{\rm succ}^{\tilde{\Theta}}(\rho(p)) = \frac{1+R_{C_{1}}(\rho(p))}{d} = \frac{1+p(d-1)}{d}.
\end{equation}
This can be input into the lower bound to $R_{C_{k}}(p(\rho))$ in Eq.~(14) in the main text, for which we see that
\begin{equation}
\max \left\lbrace 0, \frac{1+p(d-1)}{k}-1 \right\rbrace \leq R_{C_{k}}(\rho(p)).
\end{equation}
In fact, it can be seen from Eq.~\eqref{Eq:RCkNoisyMaximallyCoherent} that this lower bound is tight.

\subsection{Experimental Imperfections}
Recall, that the phase discrimination task witnesses the robustness of multilevel coherence in a quasi-device independent way, relying only on the mild assumption that the device leaves incoherent states unperturbed. In practice, this assumption is not exactly satisfied, since experimental imperfections in general lead to unitaries that do not leave incoherent states exectly invariant. To take this into account, we replacing the upper bound $p^\Theta_{\max}$ on the incoherent success probability with an upper bound $p_{\rm{succ}}^{\Theta}(\delta)$ for the probability of success for discriminating the elements of the ensemble $\{p_m,\Lambda_m[\delta]\}$, where by $\Lambda_m[\delta]$ we denote the image of a incoherent state $\delta$ under the action of a map $\Lambda_m$ which describes the approximate application of a phase $\phi_m$.

In general, given an ensemble $\{p_m,\rho_m\}_{m=1}^n$ of $n$ possible states in which a system may be prepared, the optimal probability of guessing the actual state is given by $\min\{p|\rho_{AB}\leq p \I_A\otimes\sigma_B, \sigma\textrm{ a normalized state}\}$, with the classical-quantum state $\rho_{AB}=\sum_m p_m\proj{m}\otimes\rho_m$~\cite{Konig2009}. In our case $\rho_m=\Lambda_m[\delta]$. We will assume that the maps $\Lambda_m$ do not modify an incoherent state $\delta$ too much; more precisely, in terms of trace distance $D(\Lambda_m[\delta],\delta):=\frac{1}{2}\|\Lambda_m[\delta]-\delta\|_1\leq \epsilon$, for all $m$. This means that $\Lambda_m[\delta] \leq  \delta + \Delta_m$, with $\Delta_m\geq 0$ and $\Tr(\Delta_m)\leq \epsilon$, which in turn implies that $\Lambda_m[\delta] \leq  \delta + \epsilon\I$. It is immediate to check that then $\sum_m p_m \proj{m}\otimes \Lambda_m[\delta] \leq \big(p_{\max}(1+\epsilon d)\big)\I\otimes\sigma$, with $\sigma = \frac{\delta + \epsilon\I}{1+\epsilon d}$ a normalized state. This proves that $p_{\rm{succ}}^{\Theta}(\delta)\leq p^\Theta_{\max}(1+\epsilon d)$. We can give a reasonable estimate of $\epsilon$ in terms of the process fidelity of the experiment, using the Fuchs-van de Graaf inequality $D(\xi_1,\xi_2) \leq \sqrt{1-F^2(\xi_1,\xi_2)}$, with $F(\xi_1,\xi_2)$ the fidelity between two states~\cite{Fuchs1999}. Thus, we arrive at the estimate $p_{\rm{succ}}^{\Theta}(\delta)\leq p_{\max}(1+d\sqrt{1-\mathcal{F}^2_p})$, which can be substituted in place of $p^{\Theta}_{\max}$ in Eq.~(14) in the main text. In the case of our experiment with process fidelity $\mathcal{F}_p\approx 0.9956$ and $d=4$, this means substituting $p^{\Theta}_{\max}$ with $\approx 1.375 \times p^{\Theta}_{\max}$.
%\end{widetext}

\medskip

\twocolumngrid

\bibliographystyle{apsrevfixedwithtitles}
\bibliography{MultilevelCoherence}

\end{document}